%% file: main.tex
\documentclass[sigconf, 10pt, nonacm]{acmart}
\usepackage{tikz}
\usepackage{amsmath}
\usepackage{booktabs}
\usepackage[english]{babel}
\usepackage{blindtext}
\usepackage{pifont}
\usepackage[most]{tcolorbox}
\usepackage{xcolor}
\usepackage{enumitem}
\usepackage{siunitx}
\usepackage{graphicx}
\usepackage{subcaption}
\usepackage{multirow}
\usepackage{multicol}
\usepackage{colortbl}
\usepackage{siunitx}
\usepackage{hyperref}
\usepackage{float}
\usepackage[ruled,vlined, linesnumbered]{algorithm2e}
\usepackage{listings}

\usepackage{algpseudocode}
\usepackage{makecell}
\usepackage{graphicx}
\usepackage{url}

\usepackage{amsthm}

\theoremstyle{definition}

\definecolor{LineNumberColor}{rgb}{0,0,1}

\newcommand{\systemname}{Tommy}

\newenvironment{parafont}{\fontfamily{ptm}\selectfont}{}
\newcommand{\Para}[1]{\vspace{2pt}\noindent\begin{parafont}\textbf{\textit{#1}}\end{parafont}}

\renewcommand\footnotetextcopyrightpermission[1]{} % removes footnote with conference info
\setcopyright{none}

\begin{document}

% \title{\systemname{}: Probabilistic Fair Ordering of Timestamped Events}
\title{Probabilistic Fair Ordering of Events}
% \title{One Upping Leslie, Probabilistically}
% \title{Can We \textit{Expect} To Go Beyond Lamport?}
% \author{}
% \date{}
% \author{
% Paper \#427, 12 pages body, 15 pages total
% }

\author{
\rm{\Large Muhammad Haseeb$^{\text{\scriptsize[n]}}$ \enskip
    Jinkun Geng$^{\text{\scriptsize[n][s]}}$ \enskip
    Aurojit Panda$^{\text{\scriptsize[n]}}$ \enskip
    Radhika Mittal$^{\text{\scriptsize[u]}}$ \enskip \\
    }
\rm{
    \Large 
 Nirav Atre$^{\text{\scriptsize[c][j]}}$ \enskip
 Srinivas Narayana$^{\text{\scriptsize[r]}}$ \enskip
 Anirudh Sivaraman$^{\text{\scriptsize[n]}}$ \enskip
    }
\\
  {\Large
  $^{\text{\scriptsize[n]}}$\textit{New York University}\enskip 
  $^{\text{\scriptsize[s]}}$\textit{Stony Brook University}\enskip 
  $^{\text{\scriptsize[r]}}$\textit{Rutgers University}\enskip 
  $^{\text{\scriptsize[u]}}$\textit{UIUC}\enskip
  $^{\text{\scriptsize[c]}}$\textit{CMU}\enskip
  $^{\text{\scriptsize[j]}}$\textit{Jane Street}\enskip
  }
  % \vspace{0.1in}
}

\renewcommand{\shortauthors}{}

\begin{abstract}
    \input{abstract}
\end{abstract}

\maketitle
\input{introduction}

\input{related_work}
\input{background}
% \input{drift-vs-offsets}
\input{design-overview}

\input{design-likely-happened-before-relation}

\input{design-achieving-partial-order}
\input{design-online-sequencing}

\input{evaluation}
\input{limitations}

\input{future-work}
\input{conclusion}

\bibliography{main}

\input{appendix}

\end{document}

%% file: abstract.tex
A growing class of applications depends on fair ordering, where events that occur earlier should be processed before later ones. Providing such guarantees is difficult in practice because clock synchronization is inherently imperfect: events generated at different clients within a short time window may carry timestamps that cannot be reliably ordered. Rather than attempting to eliminate synchronization error, we embrace it and establish a probabilistically fair sequencing process. \systemname{} is a sequencer that uses a statistical model of per-clock synchronization error to compare noisy timestamps probabilistically. Although this enables ordering of two events, the probabilistic comparator is intransitive, making global ordering non-trivial. We address this challenge by mapping the sequencing problem to a classical ranking problem from social choice theory, which offers principled mechanisms for reasoning with intransitive comparisons. Using this formulation, \systemname{} produces a partial order of events, achieving significantly better fairness than a Spanner TrueTime-based baseline approach.

%% file: introduction.tex
\section{Introduction}

Sequencers are fundamental building blocks in distributed systems, providing a mechanism to impose a total order on events that may occur at different locations. They underpin many core protocols, including consensus and concurrency control. In classical consensus protocols such as Paxos~\cite{paxos} and Raft~\cite{raft}, a leader implicitly acts as the sequencer, both determining a total order of operations and coordinating agreement on that order. More recently, network-based sequencers have been proposed to offload sequencing responsibilities from higher level protocols. Systems such as NOPaxos~\cite{nopaxos}, Hydra~\cite{hydra}, NeoBFT~\cite{sigcomm23-neobft}, and SwitchBFT~\cite{nsdi26-switchbft} explicitly decouple sequencing from other protocol logic, introducing dedicated sequencer components to improve efficiency.

At a high level, the role of a sequencer is straightforward: assign ranks to incoming messages to establish a total order for processing. In most existing systems, this ranking is independent of when a message was generated. Instead, messages are ordered according to when they are observed by the sequencer, effectively implementing a FIFO policy. For many traditional applications, this behavior is sufficient, since the system only requires \emph{some} ordering, even if that ordering is arbitrary.

A growing class of applications requires a stronger notion of ordering, which we refer to as fair ordering. Unlike FIFO ordering, \emph{fair ordering requires that an event generated earlier (in terms of wall-clock time) be processed before an event generated later}. FIFO ordering approximates this goal only under limited deployment circumstances e.g., specially engineered networks deployed by high-frequency trading exchanges where clients are connected to the server with meticulously equalized data-path lengths. However, generally FIFO ordering is not a proxy for fair ordering of events. 
% \notepanda{This discussion about FIFO approximating fair ordering is a bit strange to me. The problem is not reordering (which most people will think of as packet reordering, e.g., within a flow or path), but rather that events are generated from different clients, take different (potentially unequal) paths to the sequencer, and such. Even if they took the same path, you are subject to the OS processing them in receive order. I would rephrase this to say something like `In some limited circumstances, where clients and sequencer are carefully deployed and event rates are limited, FIFO can approximate fair ordering. However, this is not true if request rates increase, as is the case with most applications, or if deployment assumptions are violated , as is the case in the cloud.'}\as{Agree with Panda. I would also quantify small and high request rates.}. However, many modern applications generate large volumes of events within extremely short time intervals, making FIFO ordering an inadequate proxy for fairness. 
% \todo{a motivating plot showing FIFO's departure from fair ordering.}

The need for fair ordering is particularly pronounced in financial exchanges, ad exchanges, and other competitive systems~\cite{adex, cloudex, dbo, shoebot1, shoebot2, bot3, aerial_marketplace}, where ordering decisions directly affect economic outcomes. We refer to such systems as \emph{auction apps}. In these environments, hundreds of clients can generate millions of events within a narrow time window following a sensitive trigger, such as a market moving announcement that is broadcast simultaneously to all participants~\cite{cloudex, dbo, jasper, onyx}. In this setting, ensuring that messages generated earlier are processed before those generated later is essential for preserving fairness. The absence of a practical fair ordering primitives is a key reason why many financial exchanges continue to operate in private and specialized data centers incurring high costs, rather than using cheaper public cloud infrastructure~\cite{onyx}. 

\begin{figure}[!t]
    \centering
    \includegraphics[width=0.5\textwidth]{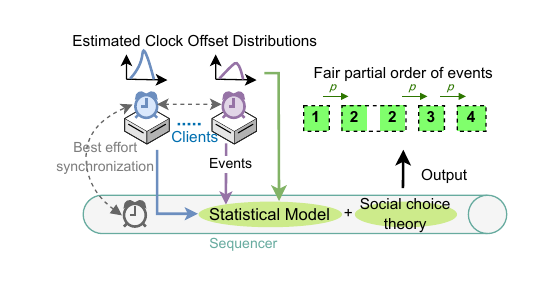}
    \vspace{-1cm}
    \caption{\textmd{The sequencer, \systemname{}, uses clock correction distributions \& timestamps to achieve a fair ordering of events.}}
    \label{fig:sys_model}
\end{figure}

Prior work has recognized the need for fair ordering in auction apps, but existing solutions either rely on strong and often impractical assumptions, such as near perfect clock synchronization~\cite{cloudex, jasper}, or are tightly coupled to the logic of a specific application, limiting their generality~\cite{dbo}.  In this paper, we describe a general abstraction for fair ordering, which we term a \emph{fair sequencer}. A fair sequencer ensures that, with high probability, an earlier event is assigned a lower rank than a later event. 
% \na{Very minor, but I found it a little bit odd that fair sequencer (a term you coined) came under "Recent Efforts". The paragraph titles you have below are very helpful, but this one might be skippable.}

% \as{Minor, but it might help to have a smoother transition to Lamport's clocks}
\Para{Classical Context:} Ordering of events is a well-studied problem. Lamport's foundational work~\cite{leslie_ordering} introduced the \emph{happened-before} relation ($\rightarrow$). If two events are causally related, meaning one can influence the other, then they can be ordered using this relation. The happened-before relation is transitive, allowing a set of pairwise causally related events to be partially ordered. Events that are concurrent, for which no causal relationship can be established, remain unordered. We generalize the ordering relation to account for such events to achieve fair ordering. 

% We use event timestamps from clients along with clock corrections distributions to probabilistically establish an event order, thus ordering even those that are not causally related.
% \na{Ah ok, I had a similar sentiment here: the interleaving between old/new ideas seems to break the flow. I think you can describe just Lamport here, and move this last sentence (pretty much as-is) to "Our Approach". All you want to convey at this point is: (a) there was this classical approach (happened-before) that enabled partial ordering of events; (b) unfortunately the classical approach doesn't work for fair ordering because there is fundamental, inherent uncertainty in timestamp acquisition. Everything beyond that is your contribution, so it's important that delineation is clear from the paragraph title.}
% order of events even if they are not causally related in the classical sense. 
% Our focus is precisely on these concurrent events. Ordering them is inherently difficult, yet crucial for achieving fairness in the applications we consider.
% \greencheck \notepanda{It is a bit unintuitive that we use timestamps, assume clocks are not synchronized, but know that $\rightarrow \subseteq \xrightarrow{p}$. It is fine if that is the case, but the statement `Our focus is precisely on these concurrent events' weirded me out as I read ahead.}

% \as{Related to Panda's comment from before. The problem is more general than "network reordering". Maybe use a different term? How about "even when events are generated by clients with differing travel times to the sequencer"}
\Para{Fundamental Challenge:} Achieving ideal fair ordering would require perfect clock synchronization, so that timestamps of the events occurring at different clients can be directly compared without being concerned about the data-path length experienced by the messages carrying information about events. However, perfect clock synchronization is impossible in asynchronous or bounded synchronous networks~\cite{limits_on_clock_sync, lundelius_clock} due to fundamental uncertainty in communication delays. In particular, it is impossible to synchronize the clocks of $n$ processes more tightly than $U(1 - 1/n)$, where $U$ denotes the uncertainty in link delays~\cite{lundelius_clock}. This inherent limitation makes achieving fair ordering a challenge.

% \Para{An Approximate Solution and Its Limits:} In environments where the time resolution of interest is much coarser than clock synchronization errors, a simple fair sequencing strategy suffices. One can wait until at least one message has been received from every client and then repeatedly release the message with the smallest timestamp, assuming in order delivery from each client. This approach yields a fair total order when clock errors are negligible. Such conditions can be achieved within a single data center, where clock synchronization errors can be reduced to nanoseconds~\cite{huygens}, making the approach viable for microsecond scale systems. However, in settings that demand finer resolution, or in multi data center deployments where clock errors can reach tens of microseconds, this strategy breaks down. These limitations motivate the need for a more general fair sequencing mechanism.

\Para{Our Approach:} 
% \na{Right now this paragraph is very matter-of-fact; your ideas are \textit{really} exciting, and I think this needs a little bit of flair to surface that. You might say something like...}\haseeb{I have rephrased. please let me know if you have something else in mind} 
We build on the insight
% \as{Provide more intuition for this observation? It's not obvious. The sentence afterwards "clients can estimate" reads well.}\haseeb{its mainly an insight, just a common statistics insight, not an observation per se, re-worded}
that \emph{two noisy timestamps from two clients can be meaningfully compared if the distributions of their clock offsets relative to a reference clock are known}. 
% Clients can estimate these distributions by collecting synchronization probes from a standard clock synchronization protocol. The distributions are shared with the sequencer, enabling probabilistic comparison of timestamps. 
% We isolate the effects of clock drift-related errors and offset-related errors and show how the drift-related errors are easily handled by typical clock sync. protocols while the offset-related errors cannot be reduced beyond a strict bound~\cite{lundelius_clock}. 
% \systemname{}, thus mainly deals with offset-related errors in the event timestamps. 
% Figure~\ref{fig:sys_model} illustrates the architecture. 
Using the insight, we introduce a new relation, \emph{likely happened before}, denoted $\xrightarrow{p}$, where $p$ represents the probability that one event occurred before another. Similar to the classical \emph{happened-before} relation, $\xrightarrow{p}$ induces a partial order, but one that also captures order fairness albeit probabilistically. 
% Ordering based on $\xrightarrow{p}$ forms the basis of fair ordering in our system,~\systemname{}. 
Figure~\ref{fig:sys_model} illustrates the architecture where the event timestamps and clock offsets distributions are shared with the sequencer that achieves a fair order of events using a statistical model and tools from social choice theory, as described below.

There are two central challenges in using the $\xrightarrow{p}$ relation for fair sequencing beyond constructing the relation for any two events. First, unlike the \emph{happened-before} relation, $\xrightarrow{p}$ is not necessarily transitive, which complicates ordering more than two events. Second, the sequencing must be performed in an online fashion i.e., a sequencer cannot assume that all events that need to be ordered have already been seen by the sequencer. An event to be seen by the sequencer in the future may need to be ordered before an event seen earlier. This complicates the sequencing in an online setting.

% \greencheck \notepanda{The structure here gets a bit strange to me, the next bold (Intransitivity) is challenge 2, the bold after is unrelated to anything (I don't think it needs to be a bold), the third bold (Sequencing in an online fashion) is challenge 3, and the fourth bold is something weird unrelated to the challenges above, but instead about eval. Where is challenge 1?}

\begin{figure}[!t]
    \centering
    \includegraphics[width=0.4\textwidth]{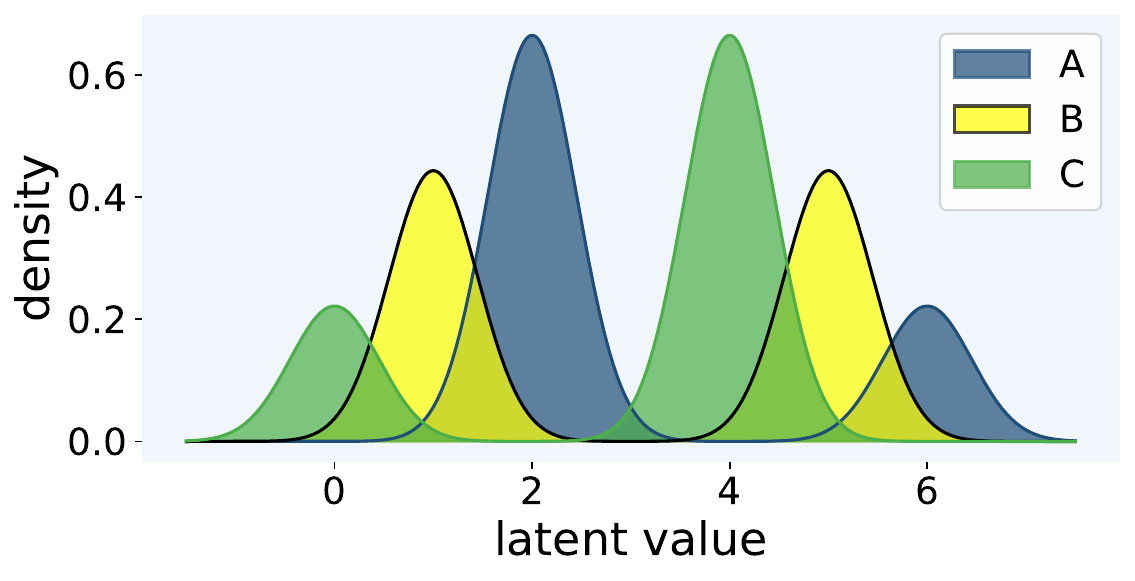}
    \vspace{-0.5cm}
    \caption{\textmd{Three distributions that have intransitive behavior: $P(A>B) > 0.5, P(B>C) > 0.5, P(C > A) > 0.5$. 
    % Distributions are mixture of Gaussians with $\sigma=0.45$. 
    }}
    \label{fig:intrans_dists}
    \vspace{-0.5cm}
\end{figure}

\Para{Intransitivity of $\xrightarrow{p}$: }Perhaps the most counter-intuitive property of $\xrightarrow{p}$ is its intransitivity. It is possible for the probability that event A precedes B to be high, the probability that B precedes C to be high, and yet the probability that C precedes A to also be high. It stems from the nature of how probabilities may behave~\cite{efron_dice}: the probability distributions may overlap in a way to show the cyclic behavior as seen\footnote{Each is a mixture of Gaussians, with standard deviation of 0.45. Centers are visible in the diagram. Known solution forms exist to calculate the shown probabilities.} in Figure~\ref{fig:intrans_dists}. 
% \as{Can you give brief intuition for this intransitivity?}
As a result, the $\xrightarrow{p}$ relation, a probabilistic relation, may exhibit intransitivity, preventing direct construction of a fair order from pairwise relations. Fortunately, social choice theory extensively deals with ranking of objects where the pairwise comparison operator can be intransitive~\cite{condorcet-book}. We reduce fair sequencing to the ranked voting problem in social choice theory to find an ordering of multiple events using the $\xrightarrow{p}$ relation. While intransitivity may still preclude the possibility of imposing a total order (due to \textit{Condorcet cycles}, a phenomenon we discuss in \S\ref{subsec:ranking-under-uncertainty}), this transformation allows us to extract a partial order in spite of adverse clock uncertainty.

% \Para{Ranking In Social Choice Theory: } A well-studied ranking problem in Social Choice Theory literature deals with finding an election winner based on pairwise preferences shown by an electorate. A candidate who beats every other candidate by majority of electorate's preferences is called a \emph{Condorcet} winner. However, such a winner may not exist because ``majority of electorate's preferences'' rule can be intransitive. In such a case, the literature proposed several solutions to find a winner(s). We adopt one such solution which ranks all candidates consistent to the comparison rule but may assign the same rank to more than one candidate. 

\Para{Online Sequencing:} Online sequencing poses challenges beyond those encountered when ordering a fixed set of events. These challenges stem from the possibility that future events may need to be assigned the same or a lower rank than events which have already arrived at the sequencer.
% already observed ones which is possible depending on the clock corrections distributions.
To resolve this, we extend our statistical model to the online setting as follows: before committing to the rank of an event, the sequencer identifies, via a statistical model, the range of timestamps for all clients that may require an equal or lower rank than the current event's timestamp; it then waits for any events that may fall within this range, and sequences them jointly with the current event to find a stable order. 

% The sequencer tracks clocks of clients via heartbeats received from them. A heartbeat with a timestamp $t$ received by the sequencer guarantees that any event generated by the respective client with timestamp smaller than $t$ must have already be seen by the sequencer (assuming clients communicate with the sequencer using an ordered-delivery channel e.g., TCP). This helps the sequencer in determining whether it has waited enough to receive events with certain timestamps from a client. 

% \Para{Evaluating such a system is hard:}
% \notepanda{Some of this sounds a bit like we are whining. I think we can put this later.}
% Evaluating such a system is fundamentally hard because of the lack of ground truth. Ground truth is the order in which all events are generated (at as well as  across clients). We resort to a simulation so that we may have ground truth. We explain how NS3 is leveraged for the simulations. 

We show that \systemname{} outperforms a Spanner TrueTime~\cite{spanner} based baseline, which represents each event timestamp as an uncertainty interval and orders events by sorting these intervals. \systemname{} produces orderings that are closer to the ideal fair order and, in some cases, exactly match it. To the best of our knowledge, \systemname{} is the first system to demonstrate that probabilistic comparisons of noisy timestamps, combined with tools from social choice theory, can yield effective fair event ordering.

% We use NS3 so that global time (simulator's time) is avaialable as a ground truth. \\

%% file: related_work.tex
\section{Related Work}
\subsection{Financial Exchanges}

\Para{Ordering in cloud-native exchanges.}
CloudEx~\cite{cloudex} and Onyx~\cite{jasper,onyx,onyx_report} are recent proposals that aim to provide fairness in trading among market participants. Several market participants may send buy or sell commands to an exchange upon some sensitive event in the market. Exchanges often offer the guarantee that whoever asked first, would get their ask met first. For approximating this guarantee, CloudEx and Onyx approximate fair ordering of events generated by clients. Both systems assume tightly synchronized clocks across all participants and the exchange server. While such assumptions are practical in special circumstances, e.g., all VMs residing in a single data center using advanced clock synchronization~\cite{huygens}, they fail to deliver desirable fairness in general settings. Generally, clock errors are large enough that two closely occurring events (e.g., microseconds apart) cannot be ordered correctly using the event timestamps from the local clocks of the respective clients. \systemname{} does not rely on tight clock synchronization; instead, it derives ordering decisions from the observed clock error distribution, which encodes more robust information for determining relative order. \systemname{} still requires some clock synchronization protocol (e.g., NTP~\cite{ntp}) so that it can build a distribution of clock corrections that is used by \systemname{}'s statistical model for computing the fair (partial) order of events with the help of tools from social choice theory. 
% Crucially, however, Tommy's statistical model can also accommodate cases where the clock synchronization error is \emph{not} negligible compared to inter-event time.

% achieves this by assuming sufficiently synchronized clocks which works fine in a single data center environment. \\

\Para{Ordering in on-prem exchanges.}
On-prem exchanges establish fair ordering by investing substantial engineering effort in bespoke infrastructure~\cite{equal_wire_length}. In particular, they ensure that all participants are connected to the exchange server using equalized data paths (e.g., same-length cables and low jitter switches). Such fine-grained engineering down to physical wire lengths is feasible in on-prem exchanges, where the entire infrastructure is purpose-built for trading. However, this approach does not generalize (e.g., to cloud environments or wide-area networks), where bespoke hardware solutions are impractical and incompatible with the needs of diverse stakeholders. \systemname{}, on the other hand, does not require any specialized infrastructure.

\subsection{Departing from Arbitrary Order}
Recent work, Pompe~\cite{osdi20-pompe}, observes that in blockchain systems, different participants have a direct stake in the specific ordering of events, and that it is therefore necessary to not be content with an arbitrary order of events produced by a sequencer. Pompe allows each replica node of a replicated state machine to provide hints about the desired order of some events. It draws on social choice theory to reframe ordering in the presence of Byzantine failures as a problem of preference aggregation under adversarial voters. In particular, it allows each participant to use physical clock timestamps as ordering indicators and collectively “vote” on a democratic ordering of events. It does not guarantee (as it is not its goal) that the final order of events will be fair in the sense that an earlier generated event will be ordered ahead of a later generated event. 

Pompe’s approach relies on the assumption that clocks at correct nodes are  accurate and treats timestamps as reliable and definitive ordering indicators. This assumption simplifies the ordering problem and avoids pathological cases such as intransitivity of ordering indicators. In contrast, in \systemname{}, we argue that clock synchronization is inherently imperfect and that timestamps (from clients' local clocks) alone cannot fully determine precedence relationships between events. Tommy incorporates this imperfect clock synchronization into its statistical model for ordering. 

%% file: background.tex
\section{Background}

\subsection{Fair Ordering}

\Para{Definition:} We consider a set of events generated at multiple clients, where each event occurs at some global time at its originating client. A sequencer observes these events and produces a total order over them. The ordering is fair if, whenever one event occurs earlier than another in global time, it is ordered before the later one. 

The definition characterizes an ideal notion of fairness that is generally impossible to realize: ordering events that occur arbitrarily close in time, but at different clients, would require resolving clients' clock errors below the fundamental uncertainty imposed by network delays, which is provably unattainable even with optimal synchronization algorithms~\cite{lundelius_clock}.
% \as{Maybe throw in a reference? I think the Lundelius and Lynch paper gets at this, but you should check the fineprint.}
% \todo{should we add some proof in appendix about impossibility or is it too obvious?}\greencheck\notepanda{I don't think its obvious, but part of the problem is that the notion of `global time' is not obvious (physics does not allow for it).}\haseeb{I have updated the wording to be more self-contained, without using the word global time. it might be ok now?}
Considering this, We relax the target and define \emph{probabilistic fair ordering}. 

\Para{Probabilistic Fair Ordering: } If an event occurs earlier than another according to the global time, the sequencer orders the earlier event before the later one with high probability. There are no guarantees on how high the probability can be. 
% In some cases, the probability may be closer to 1 while in other cases, it cannot be higher than 0.
% \as{The "in some cases" is confusing. What's the point of that sentence? Or the previous one? The first sentence makes sense.} \haseeb{removed the sentences. wanted to convey that we dont provide mcuh guarantees but i think that is laready conveyed}

Our sequencer, \systemname{}, achieves probabilistic fair ordering. It assigns ranks to all the observed events, where it attempts to assign a higher rank to events that occurred later and lower ranks to events that occurred earlier. It may assign same ranks to some events if it is not confident about their relative order. In our experiments, we show \systemname{} outperforms baselines substantially. 

\Para{Failure Model:} We assume a fail-stop model. As we explain later, each client sends heartbeats to the sequencer periodically. Absence of heartbeats from a client for some period of time is a signal for the sequencer to exclude that client from consideration. As it is a standard mechanism deployed widely, we do not discuss it further. Moreover, an interesting model to explore in the future is Byzantine failures. However, for now we believe the problem is interesting and complex enough without assuming Byzantine failures.

\vspace{-0.5cm}
%TODO: Anirudh to resume here.
\subsection{Applications requiring fair ordering}
Fair ordering is widely regarded as a fundamental requirement for many practical systems, including financial exchanges~\cite{cloudex,dbo, jasper,onyx,cme}, online gaming~\cite{syncms-2002}, and ad auctions~\cite{adex}. In such systems, fairness requires that clients have equal access to the exchange server: the order in which client writes are applied should reflect the order in which clients initiated those writes, rather than differences in network latency between clients and the server.
% \greencheck\notepanda{I think it would be good to precisely define what `fairness' means. The current statement (with a preamble of characterization sounds a bit odd, and as I indicate next, it is a bit hard to connect this to some of the consequenses.}\haseeb{subsection added in the background} Clients perform some action e.g., bidding on an asset offered by the exchange server. The server needs a way to decide which clients asked first so that it can fairly disburse assets if the demand exceeds the supply e.g., only one unit of an asset is available to be sold while multiple buyers are interested and the exchange offers the service that whoever asked first will get their ask met first. 

To achieve fairness, a straightforward approach is to rely on synchronized clocks across the exchange and all participants. Each participant adds a timestamp, via a trusted component, to their order, before sending it to the exchange. The exchange sorts incoming orders using the attached timestamps to achieve fair ordering.
% their orders sent to the exchange which sorts the incoming orders based on timestamps to achieve fair ordering. \greencheck\notepanda{I do not know why equal opportunity + timeliness $\implies$ synchronized clock. Equal opportunity sounds like a fairness requirement, which we know even holds in full asynchrony. Timeliness of information dissemination sounds like a requirement on message delay. Neither particularly sound like things that need synchronized clock above. Therefore, making it clear would be better IMO.}\haseeb{I have removed the two dimensions story of fairness as its irrelevant. only one dimension is important and i have extended the above para to explain more.}
Indeed, several recent systems~\cite{cloudex,jasper,onyx} adopt this design to build fair-access exchanges in cloud environments. However, the effectiveness of such approaches depends on near-perfect clock synchronization. Although recent advances have significantly improved clock synchronization~\cite{huygens,aws-clock-sync,nsdi22-graham}, it remains unrealistic to generally assume \emph{perfectly/tightly synchronized clocks}. In fact, classic results by Lundelius and Lynch~\cite{lundelius_clock} show that no algorithm can synchronize the clocks of $n$ servers more closely than $U(1 - 1/n)$, where $U$ denotes the uncertainty in message delay (i.e., the difference between the maximum and minimum message delays). Consequently, under high-frequency and high-concurrency workloads, even microsecond-scale clock skew can lead to ordering errors at the exchange. For example, Client~1 may generate an order 1~\textmu s earlier than Client~2, yet the exchange, after receiving the orders from both clients, may incorrectly prioritize Client~2’s order over Client~1’s and may sell a scarce asset to Client~2 incorrectly inferring that the Client~2 generated the order first.
% \greencheck\notepanda{Reinforcing the point above: but neither fo the `two dimensions' in the previous paragraph are violated by Client 2 being ordered before Client 1 in this case: both have equal opportunity (both run a request each), and neither has more information than the other.}

Motivated by these limitations, we depart from approaches that rely on near-perfect clock synchronization and instead pursue an approach that works with best-effort synchronized clocks, with the caveat that the the ordering would be fair probabilistically i.e., it outputs an order which it deems \emph{likely to be fair} given the information it has but may make mistakes. In \S\ref{sec:evals}, we show a scenario where the probabilisitic nature of \systemname{} leads it to make more mistakes than a simple baseline solution, but generally, \systemname{} outperforms the baseline. 
% \notepanda{It would help to add something saying what this means. There are many interpretations here, including WHP the order is fair (what you mean, though the use of WHP in other parts of the paper worries me a bit, given mathematically that would mean that the set of cases where it is incorrectly ordered is a set of measure 0), the probability that client feels it was sequenced fairly is high, etc.}.

% Several applications requires a mechanism to order events based on their occurrence order. Hard to achieve that order because hard to get perfectly synchronized time across clients. \\

% exchange. online-gaming/auction ad 

% JK: Merge into the section above
% \subsection{Clock sync. is the main barrier in such an ordering}
% If we had perfectly synchronized clocks, the problem becomes easier. Just account for network reordering while recieving events (cloudex, onyx type mechanism can be used), and then sort events based on their generation timestamps. \\

% But perfect clock sync is impossible. Lower bound on clock sync (Lynch paper). \\

% So we need a way to fairly order events, that;s where this paper comes into the picture. \\

\subsection{Clock Drift vs.\ Offset Considerations}
% \greencheck\notepanda{In trying to read through the paper, I am a bit unsure why this text is where it is. You need some transition to connect this to what is going on. I wonder if it might not be better to put this in Section 3, perhaps before the current 3.3 (Probability-based ranking\ldots).}

% \greencheck\notepanda{The phrase `while rendering the drift errors\ldots' worries me: the paper does not say that \systemname{} is a clock synchronization system. As such, it cannot affect drift errors, so rendering them small sounds weird. If you are a synchronization approach, you need to be much more accurate about that. Thus far it read like the system was built atop some reasonable synchronization protocol.}\haseeb{yes we are not a synchronization protocol, we just assume some stuff from the deployed protocol. re-worded.}

% \greencheck\notepanda{This implication \textbf{does not} imply (as suggested earlier) that Tommy renders drift error to 0. It merely shows that drift error can be driven arbitrarily low, but this comes at the cost of increased resource usage (to send probes). In fact, Figure 9 seems to make sense if and only if there is such a tradeoff. I would rephrase this in terms of trade-offs or something.}

Clock errors have two major components: drift and offset. Drift corresponds to rate mismatch between clocks while offset denotes the instantaneous mismatch between the values of two clocks. 
Most clock synchronization protocols can effectively correct for clock drift and reduce it to reasonable levels, whereas offset correction faces fundamental limitations.
Drift can be reduced either by increasing synchronization frequency or by explicitly estimating clock frequency.
For example, in NTP-style protocols, if clocks have fractional drift bounded by $\rho_{\max}$ and probes are exchanged every $W$ seconds, the worst-case relative drift accumulated between two clocks (in the $W$ second period) is bounded by
\[
\delta_{\text{drift}} \le 2\rho_{\max} W ,
\]
which decreases linearly as the probe interval $W$ decreases.
Similarly, protocols such as Chrony~\cite{chrony} estimate oscillator properties to infer clock frequency and continuously adjust clock readings to compensate for drift.
% , reportedly reducing it to $\approx$0.25ppm.

In contrast, offset uncertainty cannot always be reduced sufficiently.
Lundelius and Lynch~\cite{lundelius_clock} show that in a network with one-way message delay uncertainty
\(
U = d_{\max} - d_{\min},
\)
no clock synchronization algorithm can guarantee a maximum clock error smaller than
\[
\delta_{\text{offset}} \ge U\!\left(1 - \frac{1}{n}\right),
\]
where $n$ is the number of participating clocks, $d_{\max}$ is the maximum possible delay and $d_{\min}$ is the minimum possible delay.
This bound holds even in the absence of drift.

The best-case scenario (lowest offset error) occurs when $n = 2$, (excluding the trivial case of $n=1$), for which the lower bound reduces to
\(
\delta_{\text{offset}} \ge \frac{U}{2}.
\)
Even in this minimal setting, offset uncertainty remains non-zero and independent of the probe frequency and in practice, significantly larger than the drift-induced error.

\Para{Concrete example.}
With conservative drift $\rho_{\max}=20$ ppm and synchronization probes every $W=\SI{20}{\milli\second}$, the worst case drift induced misalignment is $\delta_{\text{drift}} \le 2 \cdot 20 \times 10^{-6} \cdot 0.02 = \SI{0.8}{\micro\second}$. In contrast, with one way delay uncertainty $U=\SI{40}{\micro\second}$, even two clients incur an offset uncertainty of at least $\delta_{\text{offset}} \ge U/2 = \SI{20}{\micro\second}$. Thus, frequent probing reduces drift error below \SI{1}{\micro\second}, while offset uncertainty remains an order of magnitude larger and cannot be ignored at microsecond granularity. 
% Consider clocks with conservative drift $\rho_{\max} = 20$~ppm, synchronized using probes every $W = \SI{20}{\milli\second}$ (i.e., high frequency).
% The maximum drift-induced misalignment between two clocks is
% \(
% \delta_{\text{drift}} \le 2 \cdot 20 \times 10^{-6} \cdot 0.02
% = \SI{0.8}{\micro\second}.
% \)
% Now consider a network with one-way delay uncertainty $U = \SI{40}{\micro\second}$.
% Even in the best case of only two clients, the offset uncertainty is lower bounded by
% \(
% \delta_{\text{offset}} \ge \frac{\SI{40}{\micro\second}}{2}
% = \SI{20}{\micro\second}.
% \)
% Thus, with frequent synchronization probes, drift contributes less than \SI{1}{\micro\second} of error, while offset uncertainty remains an order of magnitude larger. An application operating at microseconds granularity can safely ignore drift-induced error in this case, but not the offset-induced error. 

\Para{Implication:} Accordingly, we assume that deployed clock synchronization protocol adequately account for drift-related errors and makes it negligible w.r.t. the time granularity at which the application operates. And, we focus on reasoning about offset-related errors, the unavoidable component, in our fair ordering problem. In evaluation, we measure the impact of drift error on the fairness.
% larger drift error leads to deteriorated fairness achieved by \systemname{}. 

\subsection{Ranking Under Uncertainty and Social Choice Theory}\label{subsec:ranking-under-uncertainty}

% \greencheck\notepanda{The previous two sentences repeat what was said above.} 
Given two events $e_1$ and $e_2$ with generation timestamps  $t_1$ and $t_2$, respectively from the clients' local clocks, it may be incorrect
% \notepanda{We have not defined a safety property, and this statement thus seems a bit weird. Fairness is usually not treated as a safety property. But more importantely, it is not clear to me it is `unsafe', it seems more that you cannot conclude that $t_1 < t_2 \implies o_1\rightarrow o_2$, but that doesn't seem to make it safe or unsafe.} 
to conclude that $o_1$ should be ordered ahead of $o_2$ solely because $t_1 < t_2$, because of the clock synchronization errors. 
% \greencheck\notepanda{The previous 3 or 4 sentences seem to all repeat the same information: they are not presenting a logical sequence but rather restating the same information with different words.}
Clock synchronization errors must be explicitly accounted for when determining the relative order of events. Motivated by this observation, we develop an approach that ranks events by jointly considering their generation timestamps and the probability distribution of clock corrections/offset adjustments. 

A simple way to account for clock synchronization errors is to construct an uncertainty interval around each event timestamp, inspired by Spanner TrueTime~\cite{spanner}. And then events could be ordered by sorting these intervals i.e., events corresponding to non-overlapping intervals can get distinct ranks and those corresponding to the overlapping intervals can be given the same rank to account for high uncertainty. This is indeed our baseline ordering technique that we use to compare against \systemname{}. Although a good starting point, this approach is less helpful if the gaps between occurrence of events are small enough and the clock errors are comparatively larger as it may lead to most events having overlapped uncertainty intervals. In such a scenario, it becomes incumbent to zoom into the overlapping uncertainty intervals and attempt to assign distinct ranks to the events based on where the most \emph{uncertainty mass} lies in the intervals. This is the precise motivation for our probabilistic fair ordering.  

% Our approach relies on the assumption that clock synchronization errors are \emph{learnable}\notepanda{You need to follow this up with a phrase defining what `learnable' means, e.g., `, that is, a program can determine the distribution of errors'.}. We argue that this assumption\notepanda{Hard to follow the argument unless the reader knows what the assumption is.} is realistic in practice, as common clock synchronization protocols—such as NTP~\cite{ntp} and Huygens~\cite{huygens}—continuously estimate synchronization errors and adjust clock frequencies accordingly to minimize drift. Consequently, in our ranking approach (\S\ref{sec:model-free}), clock errors can be described using a learned error distribution. When comparing timestamps, this distribution is incorporated to derive a \emph{probabilistic} ordering relation, i.e., \emph{likely-happened-before}, between orders, rather than a deterministic one.

While the probabilistic ordering more accurately reflects the underlying uncertainty all the while enabling a useful sequencer, it also introduces cases where a total order of events cannot be established. A probabilistic relation to order two events can lead to cycles as such a relation is not necessarily transitive. 
For example, given three orders $o_1$, $o_2$, and $o_3$, it is possible that $o_1$ precedes $o_2$ with high probability, $o_2$ precedes $o_3$ with high probability, and yet $o_3$ precedes $o_1$ with high probability. In this case, pairwise comparisons fail to yield a globally consistent ordering.
Such situations correspond to \emph{Condorcet cycles} in social choice theory~\cite{condorcet-book}.
% \greencheck\notepanda{I am curious why you stick to Concordet's paradox, which is about voting systems vs. Arrow's paradox, which is the generalization. There isn't voting in the ordering bit here, so I am a bit confused/surprised.}\haseeb{arrow's paradox is that every voting system must fail somewhere. condorcet paradox is about how one common voting system fails i.e., by leading to cyclic behavior etc. So I choose to go with condorcet as I can establish parallels between $\xrightarrow{p}$ and the pairwise preference relation used by condorcet and how both these lead to cycles.}

In a ranking problem in social choice theory, deciding the winner of an election based on pairwise preferences aggregated from an electorate (where each voter states which candidate they prefer among a pair of candidates) may lead to cyclic behavior i.e., majority may prefer candidate A over B, B over C and C over A, barring the selection of a clear winner. A winner in this problem is called a Condorcet winner and it is not guaranteed to exist. The diagnosis of this problem is that the pairwise preference relation used in ranking is not necessarily transitive, similar to the pairwise probabilistic ordering relation. 

To address this challenge, social choice theory proposes the notion of the \emph{Smith set}~\cite{smithset}, which groups together all alternatives involved in Condorcet cycles such that every candidate in the Smith set beats every other candidate outside the set in a pairwise comparison, while the candidates within the set may or may not win over each other. Then the outcome of the election system is a set instead of single (Condorcet) winner. 

Inspired by this concept, we design our sequencing algorithm to cluster events into Smith sets and enforce a deterministic ordering across sets. Within each Smith set, we intentionally refrain from imposing a strict order; instead, we delegate the choice of intra-set ordering to the application, allowing flexibility in how ties or near-ties are resolved, leading to a partial order instead of a total order produced by the sequencer. 

% Another background important to understand the paper: social choice theory's method of finding election winner based on pairwise comparisons. A Condorcet winner may not always be possible, a smith set notion is needed. We use that. \\

%% file: design-overview.tex
\section{\systemname{} Sequencer Design}
% \greencheck \notepanda{Calling this an Overview makes me believe there is another section which has detailed design, which there is not. How about retitling this section `\systemname{} Sequencer' or something.}

% Figure~\ref{fig:sys_model} illustrates the \systemname{} architecture. \systemname{} depends on a best effort clock synchronization protocol. Although \systemname{} embraces that clocks can never be perfectly synchronized~\cite{lundelius_clock}, but a loose/best effort clock synchronization is at the heart of estimating clock corrections distributions. All client machines and the sequencer run such a clock synchronization protocol in which clients synchronize their clocks to the sequencer’s clock which is assumed to be the global clock. \greencheck \notepanda{This sentence comes a bit out of nowhere: the previous sections have been trying to convince the reader synchronized clocks are not good, but here we introduce the need to synchronize (I understand why, but just emulating a reader). It would be good to start off by having a sentence such as \systemname{} uses almost synchronoized clocks, or something to ease the reader in.}
As a clock synchronization protocol runs, it adjusts each client's clock periodically. These adjustments constitute the correction distributions for each client. 

As an event $e$ occurs at a client $c$, the client generates a timestamp, $t_e$, from its local clock to denote when the event occurred. The event timestamp, $t_e$, may not exactly represent the moment $e$ occurred as there is some lag, however minute, between the occurrence of $e$ and generating the $t_e$ in any real-world system. We ignore this lag. 

Clients convey the information about events to the sequencer by sending messages that contain event timestamps. A client $c$ ensures that the message for $t_e$ is not sent after the message for $t_{e'}$ if $t_e < t_{e'}$ and $e, e'$ are generated at $c$. The sequencer uses the event timestamps as well as the clock correction distributions to assign ranks to the events that constitutes probabilistic fair ordering. 

The sequencer uses a statistical model
% \greencheck\notepanda{Is `us' the sequencer, in which case just say `The sequencer uses a statistical model to \ldots'} 
to find the probability of one event $e$ being earlier than another event $e'$ w.r.t. the global clock by only comparing the event timestamps $t_e$ and $ t_{e'}$.
% \greencheck\notepanda{The term `local timestamp' is a bit confusing. In most contexts, `local timestamp' is the timestamp at the node, in this case the sequencer's local time. What you mean is the timestamp carried in the message. It might be good to rename it, perhaps to the `message timestamp' (or `event timestamp').}\greencheck\notepanda{The conflation between messages and events is going to make it hard on a reader. I strongly suggest picking one term.}. 
This is enabled by treating global time as a random variable. Once the pairwise probabilities for every two events are calculated, the sequencer uses a social choice theory based solution to deal with intransitivities and find a partial order of all events. 

\Para{Likely-happened-before relation:} The foundation of fair ordering is a pairwise relation, which we term likely-happened-before, denoted as $\xrightarrow{p}$, where $p$ represents the probability that one event occurred before another. \systemname{} models global time at which the event occurred as a random variable: for an event $e$, global time of event occurrence, $T_e$ is expressed as a linear combination of the event timestamp, $t_e$, from the respective client's local clock and the client's clock corrections distribution. By comparing realizations of these random variables for two events, the sequencer computes the probability that one event precedes the other. This probability directly defines the $\xrightarrow{p}$ relation.

Constructing the $\xrightarrow{p}$ pairwise relation is the key gadget that enables \systemname{}’s fair ordering. In contrast, attempting to eliminate all synchronization errors from each client’s local clocks to achieve fair ordering (by simply ordering events based on sorted timestamps) poses a substantially more difficult problem than probabilistically comparing two global timestamps.

% Using clock correction distributions to pose global time as a random variable. Finding one random variable smaller than the other an easy problem compared to removing all the error from client's local clocks. \\

% Finding pairwise probabilities for N messages can become expensive. We propose an efficient approach leveraging the insight that clock error distributions are stable on short time horizons. \\

% Algorithm \ref{alg:pairwise-points}. But we do not use it and replace with a more efficient version as explained later. 

\Para{Ordering multiple events using the $\xrightarrow{p}$ relation:} Unlike the classical \emph{happens-before} ($\rightarrow$) relation, the $\xrightarrow{p}$ is not  a transitive relation. So obtaining an order of more than two events using the pairwise $\xrightarrow{p}$ relation is a challenge. We map the problem of finding order of events using the pairwise relations to a similar problem in the social choice theory literature: ranking candidates in an election using pairwise preferences (i.e., Candidate A is preferred over candidate B) obtained from an electorate. The reduction of ordering of events problem to a well-known ranking problem lends us mature tools that explicitly accounts for possible intransitivity in the pairwise preferences/ordering. 

% Algorithm \ref{alg:find-order}. 

% Given pairwise comparisons, finding order of several events is not straightforward as the comparisons are probabilistic and thus, the constructed likely-happened-before relation is not necessarily transitive. \\

% Solution to intransitivity inspired from Social choice theory: repeated smith set extraction. \\

% Each SCC/smith set gets one rank, while successive SCCs gets higher ranks.  \\

\Para{Online Sequencing: } As a sequencer receives events from clients, it needs a mechanism to determine when the current ordered set of events is stable enough to be emitted (to any downstream application). We consider a set of events stable if the probability that any future event would arrive and be ordered ahead of (or equivalently, interleaved before) any event in the set is sufficiently low. Furthermore, achieving high performance as events arrive in a stream requires processing incoming events with minimal latency, which in turn demands efficient implementations of both the statistical model and the social choice theory–based ranking mechanism. We propose a \emph{probability-sketch} based efficient algorithm for finding probability of an event to have occurred before another. Our algorithm leverages the observation that clock-correction distribution of a client does not change significantly at short timescales. In the subsequent sections, we expand on how the pairwise probability is calculated to construct the $\xrightarrow{p}$ relation (\S\ref{sec:model-free}), how the relation is used to get an ordering of several events (\S\ref{sec:graph}) and how the sequencer operates in an online fashion (\S\ref{sec:online}). 

% Sketch of clock error.

% Algorithm \ref{alg:precompute-diffs} and \ref{alg:pairwise-query} because Algorithm \ref{alg:pairwise-points} is too expensive to run in an online fashion. 

% To counter network variability, sequencer waits for one message/heartbeat from each client, orders them and finds out whether the ordering is stable by finding the probability that any new message will not need an equal or lower rank.  \\

%% file: design-likely-happened-before-relation.tex
\section{Constructing Likely-Happened Before, $\xrightarrow{p}$, Relation}
\label{sec:model-free}
\label{likely-happened-before}

% \greencheck \notepanda{It is strange to see a section about a relation (a mathematical object that exists independent of whether it can be computed or not) start with the words `System Model', and hten a siscussion of TCP. Do you perhaps mean `Computing LHBR'?}

An event $i$ with event timestamp $t_i$ from the local clock of the respective client, $c(i)$, has the global timestamp $T_i$. Under ideally synchronized clocks, $T_i = t_i$. However, in the real-world, every clock synchronization protocol is imperfect. We represent the global timestamp as:
\(
T_i = t_i + \theta_i,
\)
% \greencheck\notepanda{I do not know how to interpret the next sentence, nor how it connects to anything.}
where $\theta_i$ denotes the client’s clock correction relative to the sequencer's clock (the global clock) at the instant event $i$ is generated.
% \greencheck\notepanda{and $T_i$ denotes\ldots}. 
The correction $\theta_i$ is unknown but is modeled as a random variable following a probability distribution $f_{\theta_i}$.

Clock correction distributions may differ across clients due to heterogeneous synchronization conditions, such as temperature variation across a data center or asymmetric network latency~\cite{clock-temperature-report,huygens}. Each client estimates its own correction distribution by collecting clock synchronization probes. A clock synchronization probe is an adjustment made to a client's clock by a clock synchronization protocol. A sequencer can observe $t_i$ for each event and $f_{\theta_i}$ for each client. 

% Finding pariwsie probability, the math goes here. This probability is the basis for likely-happened before relationship.  \\

\Para{Comparing two global timestamps:} It is impossible to compute $T_i$ exactly but we can compare two global timestamps $T_i$ and $T_j$ by only observing $t_i$ and $t_j$ using a probabilistic analysis that assumes the knowledge of clock correction distributions $f_{\theta_i}$ and $f_{\theta_j}$.

We analyze the probability that one event precedes another. This probability is called the \emph{preceding-probability}:
\[
    \mathbb{P}(T_i < T_j \mid t_i, t_j) = \mathbb{P}(t_i + \theta_i < t_j + \theta_j).
\]
Rearranging:
\(
    \mathbb{P}(T_i < T_j \mid t_i, t_j) = \mathbb{P}(\theta_i - \theta_j < t_j - t_i).
\)
Since $\theta_i$ and $\theta_j$ are random variables, their difference follows a new distribution:
\[
    \Delta \theta = \theta_i - \theta_j \sim f_{\Delta \theta}.
\]
Then the preceding-probability is given by:
\[
    \mathbb{P}(T_i < T_j \mid t_i, t_j) = f_{\Delta \theta}(t_j - t_i) = \int_{-\infty}^{t_j - t_i} f_{\Delta \theta} d\Delta.
\]

If $f_{\theta_i}$, $f_{\theta_j}$ are modeled as any known distributions e.g., Gaussian, then we may get a known solution form for the above integral.
% \greencheck\notepanda{The following phrasing is scary without additional information. Leads to worse fairness in what circumstances? Is it a universal? Is it because of how experiments were setup? What leads to this?}
We have explored modeling $f_{\theta_i}$ and $f_{\theta_j}$ using such distributions, by fitting
these distributions to clock synchronization probes collected in our
experiments. However, we find that imposing such modeling assumptions
consistently leads to worse fairness compared to making no distributional
assumption at all.
%we find that modeling the $f_{\theta_i}$, $f_{\theta_j}$ as a Gaussian, Bimodal or Pareto distributions, by fitting a distribution to the clock sync. probes in our experiments, leads to worse fairness than making no such modeling assumption. 
Therefore, we make no assumption about the shapes of $f_{\theta_i}$, $f_{\theta_j}$. We evaluate the above integral empirically using the following estimator:

% Difference of two independent random variables smaller than x

% where the difference distribution is \(Z = c_x - c_y\)
% :

\[
\hat f_{\Delta \theta}(t_j - t_i)
\;=\;
\frac{1}{|f_{\theta_i}|\,|f_{\theta_j}|}
\sum_{\theta_i \in f_{\theta_i}}
\sum_{\theta_j \in f_{\theta_j}}
\mathbf{\textbf{1}}\!\left[\theta_i - \theta_j < t_j - t_i\right].
\]

% Intuition for the estimator:
% for all possible i,j: se if i-j smaller than x then add 1
% doing that gives us count of all events of interest
% divide by total events and you get probability

% \[
% \hat f_{\Delta \theta}(x)
% \;=\;
% \frac{1}{|f_{\theta_i}|\,|f_{\theta_j}|}
% \sum_{\theta_i \in f_{\theta_i}}
% \sum_{\theta_j \in f_{\theta_j}}
% \mathbf{\textbf{1}}\!\left[\theta_i - \theta_j < x\right].
% \]

The estimator makes an independence assumption between $f_{\theta_i}$ and $f_{\theta_j}$
% \greencheck\notepanda{And thus, the estimator is not model-free, as was claimed above. Just say you don't assume a distribution for $f_\theta$ rather than the broader model-free terminology.}. 
The assumption might not be always true but this is a simplification choice we make and may lead to sub-optimal ordering. In evaluation, we note \systemname{} outperforms existing mechanisms despite our simplification assumption. Algorithm \ref{alg:pairwise-points} implements the above estimator to find the pairwise preceding probability for two events.

\begin{algorithm}[t]
\fontsize{10pt}{10pt}\selectfont
\caption{Pairwise Ordering Probability}
\label{alg:pairwise-points}
\KwIn{
Local timestamps $x, y$; \\
Correction probes $\mathcal C_x = \{c_x^1,\dots,c_x^{n_x}\}$; probes carry synchronization protocol's adjustments to the clock\\
Correction probes $\mathcal C_y = \{c_y^1,\dots,c_y^{n_y}\}$
}
\KwOut{$p = \Pr(x + c_x < y + c_y)$}

$\tau \leftarrow y - x$ \\
$count \leftarrow 0$ \\
$total \leftarrow n_x \cdot n_y$

\ForEach{$c_x \in \mathcal C_x$}{
    \ForEach{$c_y \in \mathcal C_y$}{
        \If{$c_x - c_y < \tau$}{
            $count \leftarrow count + 1$
        }
    }
}

\Return $count / total$
\end{algorithm}

\Para{Likely-happened-before and intransitivity: }Once we can find preceding-probability for two events, this is the basis of the \emph{likely-happened-before}, $\xrightarrow{p}$, relationship. It is a pairwise relation where $i \xrightarrow{p} j$ denotes that event $i$ happened before event $j$ with probability $p$. This is the foundation for the fair ordering achieved by \systemname{}. 

If $\xrightarrow{p}$ was a transitive relation then finding a total order for a set of events would be trivial. We define transitivity as: If $i \xrightarrow{p > 0.5} j$ and $j \xrightarrow{p > 0.5} k$ then we must have $i \xrightarrow{p > 0.5} k$. However, the transitivity may not always hold. In case it was to hold, then fair ordering would simply be a topological ordering of the graph where nodes are events and an edge from $i$ to $j$ represents $i \xrightarrow{p > 0.5} j$. 

In some special cases, the transitivity for $\xrightarrow{p}$ may hold. For example, Appendix~\ref{app:guassian} shows that if clock correction distributions are Gaussian, then the transitivity provably holds. We also implement the proof in Lean. The existence of transitivity depends on the behavior of probabilities and it generally does not hold~\cite{efron_dice}. Finding the necessary and sufficient conditions under which transitivity may hold is left for future work. 

% \greencheck\notepanda{I am entirely unsure why I am reading about Efron dice here, and why it connects to the rest of it. I would drop this.}
% A famous example of intransitivity in probabilities is a special kind of dice called Efron dice~\cite{efron_dice}. Given 3 Efron dices, the following three statements are true: (i) Dice A has high probability of beating Dice B i.e., rolling Dice A and B would yield a number on the top face of Dice A that is higher than the number on the top face of Dice B , (ii) Dice B has high probability of beating Dice C and (iii) Dice C has high probability of beating Dice A. This cyclic behavior shows that we cannot rank all the tree dices to show which dice is most likely to win or which dice is least likely to win. 

Due to the intransitivity, $\xrightarrow{p}$ may not always allow a total order of events. In our ns-3 simulations with a data-center wide background workload, we see that such intransitive behavior is indeed manifested while attempting to order events based on $\xrightarrow{p}$, thereby preventing us from achieving a total order of events. Despite the existence of intransitivity, in \S\ref{sec:graph}, we will discuss how a partial order of events can still be achieved by using lessons from social choice theory. 

% Once likely-happened-before relation is constructed, show how this relation can be intransitive. Give examples of Efron dice and how multi-modality leads to such behavior. Tell that we see multi-modality in a real-world like setting i.e., NS3 DCN simulation. Intransitivity is important to mention because happened-before is transitive and allows to order more than 2 events. With likely-happened before, not so straightforward as we discuss in next section. \\

% Tell that Gaussian is too nicely shaped to lead to intransitivity, we have a Lean proof. But real world latencies does not give us Gaussian.  \\

% \greencheck\notepanda{I was a bit unsure about why this is here? It seems to be about an algorithm that was presented a sectio or more ago. But why is it separated by an ocean?}

\Para{Efficient events processing: } Pairwise preceding probabilities are a core primitive used repeatedly by the fair ordering procedures. As we discuss next, in Algorithm~\ref{alg:find-order}, constructing the matrix $P$ requires computing pairwise preceding probabilities for all $O(n^2)$ pairs of events. Using Algorithm~\ref{alg:pairwise-points} to compute the preceding probability for a single pair of messages incurs quadratic time complexity with respect to the number of samples in the clock correction distributions.

% \greencheck\notepanda{I think it would be nice to say something about the fact that this is a tradeoff: you are adding an additional assumption (stability of clock correction) to get performance.}
We observe that clock correction distributions remain stable over time
% \greencheck\notepanda{How short? How long? What happens if this assumption is violated?}\haseeb{We dont know that and it is mostly the matter of doing benchmakring in the env where the sequencer is supposed to run.  but I dont want to go into details of these things as they only matter if deploying sequencer into production and i think has less of an impact on the message of the paper. so i am using wording to avoid saying these things explicitly. does that make sense? I have added a new para after this para to briefly talk about it.} 
and only change if the background workload shifts significantly or the underlying network connectivity changes, i.e., some link failure leading to rerouting of packets through a path which may have different characteristics than the abandoned path. Assuming the distributions stay stable over timescale of $M$ minutes, the sequencer can periodically pre-compute
% \greencheck\notepanda{How does an administrator pick the periodicity?}\haseeb{Check the next para. baseically i jsut want to avoid that discussion. }
\emph{difference distributions} using Algorithm~\ref{alg:precompute-diffs} after every $M$ minutes. Once these difference distributions are available, the pairwise preceding probability for any pair of events can be computed using Algorithm~\ref{alg:pairwise-query} in logarithmic time as opposed to the quadratic time with respect to the number of samples in the clock correction distributions. 

One disadvantage of pre-computing difference distributions is that if the difference distributions become outdated, the resultant values of preceding probabilities may become incorrect leading to wrong ordering. The distributions may become outdated if the timescale of refreshing the distributions is too large which is relative to the network environment in which \systemname{} is deployed. We excuse ourselves from further discussion of this problem but may require more research if \systemname{} is to be deployed in production. 

\begin{algorithm}[t]
\fontsize{10pt}{10pt}\selectfont
\caption{Precompute Difference Distributions}
\label{alg:precompute-diffs}
\KwIn{
Clock corrections probes sequences for all clients $\{\mathcal C_1, \ldots, \mathcal C_n\}$,
where $\mathcal C_i = \{c_i^1, \ldots, c_i^{m_i}\}$
 is the sequence for one client}
\KwOut{Sorted difference lists $\mathcal D_{ij}$ for all $i,j$}

\For{$i \gets 1$ \KwTo $n$}{
    \For{$j \gets 1$ \KwTo $n$}{
        $\mathcal D_{ij} \leftarrow \emptyset$ \\
        \ForEach{$c_i \in \mathcal C_i$}{
            \ForEach{$c_j \in \mathcal C_j$}{
                Append $(c_i - c_j)$ to $\mathcal D_{ij}$
            }
        }
        Sort $\mathcal D_{ij}$
    }
}

\Return $\{\mathcal D_{ij}\}_{i,j=1}^n$
\end{algorithm}

\begin{algorithm}[t]
\fontsize{10pt}{10pt}\selectfont
\caption{Pairwise Ordering Query}
\label{alg:pairwise-query}
\KwIn{
(1) Events $e^{x}_i$, $e^{y}_j$ where $x$, $y$ are local timestamps from clients $i$, $j$.
(2) $\mathcal D_{ij}$
}
\KwOut{$p_{ij} = \Pr(x + c_i < y + c_j)$}

$\tau \leftarrow y - x$ \\
$k \leftarrow \textsc{UpperBound}(\mathcal D_{ij}, \tau)$ \tcp*{Count values $< \tau$}
\Return $k / |\mathcal D_{ij}|$
\end{algorithm}

%% file: design-achieving-partial-order.tex
\begin{algorithm}[t]
\fontsize{10pt}{10pt}\selectfont
\caption{Probabilistic Fair Ordering}
\label{alg:find-order}
\KwIn{
(1) Event log $E = \{e^t_1,\dots,e^{t'}_n\}$ where each event $e^{t'}_n$ was generated by client $n$ at local timestamp $t'$; \\
(2) $\{\mathcal D_{ij}\}_{i,j=1}^n$ (generated by Algorithm \ref{alg:precompute-diffs})
}
\KwOut{
Ordered batches of events $\mathcal{B}_1,\mathcal{B}_2,\dots$
}

\tcp{Generate all to all matrix for pairwise probabilities $P_{ij}$ for all events }
$P \leftarrow \textsc{PairwiseProbMatrix}(E, \{\mathcal D_{ij}\}_{i,j=1}^n)$

$G \leftarrow$ directed graph with vertices $\{1,\dots,n\}$ \\
\For{$i \gets 1$ \KwTo $n$}{
    \For{$j \gets 1$ \KwTo $n$}{
        \If{$P_{ij} > \tfrac12$}{
            add directed edge $i \rightarrow j$ to $G$
        }
    }
}

Compute strongly connected components (SCCs) of $G$ \\
Collapse SCCs into a DAG $G'$ \\
Compute a topological ordering of $G'$ \\
\Return SCCs in topological order
\end{algorithm}

\section{Achieving Partial Order}
\label{sec:graph}

Although $\xrightarrow{p}$ is helpful in ordering two events, it is not a general ordering relation because of its intransitivity. 
% \greencheck \notepanda{I would weaken prohibit to something else. The thing is intransivity means it is not an ordering relation at all: even partial orders need to be transitive.} 
This problem is analogous to a problem in social choice theory literature called \emph{Condorcet Paradox}~\cite{condorcet-book} briefed next. 

In a ranking problem in the social choice theory, given a set of pairwise preferences of election candidates, the goal is to find the election winner, called a \emph{Condorcet winner}. A pairwise preference is a preference of one candidate between two candidates and is analogous to the $\xrightarrow{p}$ relation where both are intransitive. A \emph{Condorcet} winner is a candidate who would receive the support of more than half of the electorate in a one-on-one race against any one of their opponents. 
% \greencheck \notepanda{Rephrase next sentence: `The paradox shows that such a winner need not exist.'}
The \emph{Condorcet Paradox} is that such a winner may not exist. However, a generalization of \emph{Condorcet} winner called a \emph{Smith Set} is guaranteed to always exist. 

A Smith Set is the smallest set of candidates that are pairwise unbeaten by every candidate outside of it. We extend the idea of finding Smith Set to get a partial order of events: (i) find a smith set from the events, (ii) assign the same rank to all events in that smith set, (iii) remove those events from the set of events and, (iv) repeat while incrementing the assigned rank and until all events have been ranked/ordered. 

The above solution provides a 
% \greencheck \notepanda{And here we come back to my comment above: you had to do this to recover even a partial ordering relation. So just start by saying intransitivity means $\xrightarrow{p}$ is not an ordering relation.} 
partial order as some events are not ranked w.r.t each other, i.e., they are assigned the same rank. 
% \greencheck \notepanda{The switch to ranking here is a bit strange, might need some help earlier in the paper. Until this point, it seemed like you were drawing an analogy to ranking in social choice theory, but here it suddenly turns to assigning ranks even for your system. This is fine, but a little pre-intro might help.}
We forsake achieving a total order and settle for a partial order. An application can have its own rules for extending the partial ordering to a total order, an already exercised option where applications decide how to deal with concurrent events when the $\rightarrow$ relation leads to a partial order. The sequencer releases batches (smith sets) of events at a time to any downstream application.
% \greencheck\notepanda{It is a bit strange to not see something to the effect that forsaking this is OK.}
An application can decide for itself how to deal with tied ranks. For example, a financial exchange may try to randomly find a bidder from a batch of bidders to sell an asset to or divide an asset proportionally to all bidders in a batch. 

The above mentioned repeated smith set extraction can be implemented as a graph problem. The sequencer creates a graph where each node is an event. An edge exists between two nodes $i, j$ if $i \xrightarrow{p > 0.5} j$. Then the sequencer runs Tarjan's Algorithm~\cite{tarjan-algo}  on the resultant graph, to find all the Strongly Connected Components (SCC) of the graph. Finally, it finds a topological ordering of the SCCs. Each SCC is a smith set and they are emitted from the sequencer in the order of their topological order. Algorithm~\ref{alg:find-order} implements the above.

% Ordering more than two events, is not easy because of intransitivity. Show a cycle, topo order not possible.  \\

% We have a similar problem in social choice theory literature: ranking candidates given votes. A solution is to find smith set, we use that. Repeatedly find smith set, remove, find next smith set and so on.  \\

% This repeated smith set extraction can be achieved by a simple algorithm: if we have a graph representing pairwise comparison results for every two nodes, do SCC decomposition and topo sort the SCCs.  \\

% How we get the graph: Use likely happened before relations to construct an all-to-all graph and remove edges with probability lower than half (assume no p=0.5 edges). This graph is a tournament.  \\

% The tournament graph can have cycles. (In case of guassian, no cycles, topo order possible). So a topo order may not be possible. But we can do SCC decomposition of the graph and get topo order of SCCs. These SCCs are smith sets. We rank SCCs and that is our fair ordering. It is partial (each SCC can have more than one node) order.  \\

%% file: design-online-sequencing.tex
\section{Online Sequencing}
\label{sec:online}
\label{sub:stable}

% \greencheck\notepanda{Proposed mod to the next two sentences: `
% Our discussion so far has assumed that the sequencer needs to order a set of available events. However, in reality, the sequencer needs to operate online, ordering events as they arrive as a stream.'}
Our discussion so far has assumed that the sequencer needs to order a set of available events. However, in reality, the sequencer needs to operate online, ordering events as they arrive as a stream. For \emph{online sequencing}, we need a way to answer the following question: should the sequencer order and emit the current events (to any downstream application) or wait for further events? The waiting would be needed if a new arriving event has a lower rank than the already present events.
% \greencheck\notepanda{Saying this as an arriving message might need to be ordered before the current message is a bit easier for me to visualize: ranks are at a remove. But this could be fixed by having the notion of ranking as the initial problem statement in the paper.} 
In such scenarios, skipping the wait may lead to a sequencer output that disagrees with the fair order of events.
% \greencheck\notepanda{This statement is strange: the `fair ordering' is not what is incorrect, rather the order output by the sequencer disagrees with the `fair order'.}

Once the sequencer knows all new events will only need a higher rank than the already arrived events, it can order the existing events and deem it stable for emitting from the sequencer. In the following, we discuss how this stability condition can be checked. 
% \greencheck\notepanda{This next statement is strange: efficiency can mean all sorts of thing, but latency is usually not an efficiency thing, resource requirements are.}

\SetKwProg{Fn}{Function}{:}{end}
\begin{algorithm}[t]
\fontsize{10pt}{10pt}\selectfont
\caption{Online Sequencing}
\label{alg:online-stable-time}

\KwIn{Number of clients $N$;}

\BlankLine
\textbf{State:}
$\mathsf{buffer} \gets \emptyset$ (received but unordered events)\;
$\mathsf{watermark}[c] \gets -\infty$ for all $c \in \{0,\ldots,N-1\}$\;

\BlankLine
\Fn{\textsc{OnEventArrival}$(e)$}{
  % \tcp{$e$ has fields: $e.\mathsf{id}$, $e.\mathsf{client}$, $e.t'$ (local time)}
  $\mathsf{buffer} \gets \mathsf{buffer} \cup \{e\}$\;
  \textsc{AttemptOrder}()\;
}

\BlankLine
\Fn{\textsc{OnHeartbeat}$(c, t)$}{
  % \tcp{Heartbeat means client $c$ has sent all events with local time $\le t$, clients generate heartbeats periodically}
  $\mathsf{watermark}[c] \gets \max(\mathsf{watermark}[c], t)$\;
  \textsc{AttemptOrder}()\;
}

\BlankLine
\Fn{\textsc{AttemptOrder}()}{
  \If{$\mathsf{buffer} = \emptyset$}{
    \Return\;
  }

  $\mathsf{batches} \gets \textsc{OrderBuffer}(\mathsf{buffer})$\;
  \tcp{$\mathsf{batches}$ is a list of batches/SCCs generated by Algorithm~\ref{alg:find-order}}

  $(t_F, x_F) \gets \textsc{Frontier}(\mathsf{batches})$\;
  \tcp{Choose frontier event in last batch with maximum local time $t_F$ and its client $x_F$}

  $\mathsf{stable\_until}[\cdot] \gets \textsc{StableTimes}(t_F, x_F)$\;

  \For{$c \gets 0$ \KwTo $N-1$}{
    \If{$\mathsf{watermark}[c] < \mathsf{stable\_until}[c]$}{
      \Return \tcp*{Not safe yet; wait for more events/heartbeats}
    }
  }

  \textsc{Emit}$(\mathsf{batches})$\  \tcp*{Safe to commit}
  $\mathsf{buffer} \gets \emptyset$\;
}

% \BlankLine
% \Fn{\textsc{Frontier}$(\mathsf{batches})$}{
%   \tcp{Return $(t_F, x_F)$ from the last batch}
%   $B \gets \mathsf{batches}[\lvert \mathsf{batches}\rvert]$\;
%   $(t_F, x_F) \gets \arg\max_{e \in B} e.t'$ \tcp*{also return $e.\mathsf{client}$}
%   \Return $(t_F, x_F)$\;
% }

\BlankLine
\Fn{\textsc{StableTimes}$(t_x, x)$}{
  \tcp{For each client $c$, find the smallest $t$ such that $\Pr(t_x + c_x < t + c_c) \approx 1$}
  \tcp{(via per-client binary search on future timestamp t)}
  \Return $\mathsf{stable\_until}[\cdot]$\;
}

% \BlankLine
% \Fn{\textsc{OrderBuffer}$(\mathsf{buffer})$}{
%   \tcp{Abstract ordering procedure over current buffer (e.g., SCC+topo on pairwise probabilities)}
%   \Return $\mathsf{batches}$\;
% }
% \vspace{-0.7cm}
\end{algorithm}

% \Para{Stable ordering:}

% As messages arrive at the sequencer, it must determine when it is safe to assign final ranks i.e., when, with high probability, no future message could require a rank equal to or lower than those already assigned. Once this \emph{stability} condition holds for a set of messages, the sequencer may proceed to order them and emit the resulting ordered messages to downstream applications.

% \todo{Add equation from HotNets here}

Consider an event $m$ from client $c$ generated at timestamp $t$. For this event, the sequencer computes a \emph{stable timestamp} for every other client. A stable timestamp $t'$ for client $c'$ implies that any event $m'$ generated by $c'$ at or after $t'$ has a very high probability of being assigned a rank strictly higher than that of $m$. Therefore, if the sequencer has received all events with timestamps smaller than the respective clients' stable timestamps, it can safely order and emit the events.

To enable this determination, each client periodically sends heartbeat events to the sequencer. A heartbeat generated at timestamp $t$ and received by the sequencer guarantees that all events generated by the sending client at timestamps $t' < t$ have already been received by the sequencer.

Using heartbeats, the sequencer can check whether all events with timestamps earlier than the computed stable timestamps have arrived. If additional events with timestamps smaller than the current stable timestamps arrive, the sequencer must re-compute the stable timestamps accounting for the new events. Algorithm~\ref{alg:online-stable-time} implements the online sequencing procedure that enforces the stability condition.
% \greencheck\notepanda{I am surprised there is not just a timeout. If a client does not send a timeout for some time, you could just assume there are no more messages with earlier timestamp. Generally, I think it would be better to say that you error out messages that arrive too late (equivalently, if a client has particularly large error), rather than saying there is no liveness.}\haseeb{The timeout is a good idea but I am trying to not add any more parameters as each parameter demands some corresponding experiments}

We show in \S\ref{sec:evals} that with large clock correction distributions, it may take longer to reach the stability condition which, in turn, delays emitting ordered events from the sequencer. Hence, any production realization of \systemname{} may add sanity checks around unusually high clock sync. errors that lead to large correction distributions and may bar the corresponding clients from partaking in the sequencing for the sake of liveness. 

% \todo{discuss how stable times calculation avoid leading to a case: t from c has probability=1 to be after current event's t'. but t is part of a cycle and some other event will loop bacl to t' because of intransitivity}

% Any clients which may have unusually high clock sync. errors and in turn, large tails in their correction distributions can be ignored until the corresponding circumstances improve. 

% \Para{Liveness: (TODO)} EMIT may never reach in the above algorithm depending on the message generation pattern and the clock correction distributions. 

% To release a set of ordered events, the sequencer needs to ensure no new events need lower or equal rank. A sequencer waits for one message/heartbeat from each client. Orders them and finds out (find if the probability is low enough for ) whether it needs to wait for new messages/heartbeats before realsing the batch.  \\

% Replacing Algorithm~\ref{alg:pairwise-points} with the combination of Algorithm~\ref{alg:precompute-diffs} and Algorithm~\ref{alg:pairwise-query} yields substantial performance improvements.

% Events arrive in a stream, finding pairwise probability of N incoming events can be time consuming.  \\

% We propose a solution noting that distributions stay stable over some time.  We precompute some stuff so that finding pairwise probability of two incomign events is cheap.  \\

%% file: evaluation.tex
\section{Evaluation}
\label{sec:evals}

For evaluating \systemname{}, we require ground truth, namely the true temporal order in which events occur, against which we can compare the order produced by \systemname{}. However, obtaining ground truth in real-world systems is fundamentally challenging, as it would require perfectly synchronized clocks, which are impossible to realize.

\Para{Difficulty of obtaining ground-truth:} One might attempt to approximate ground truth by using clocks that are significantly better synchronized than those available to \systemname{}. However, this approach is insufficient on its own. Even if two clocks with different synchronization accuracies are available, the variance (across clients) in latency between reading the two clocks can introduce additional uncertainty, thereby obscuring the true event order. Only in an idealized setting where both clocks could be accessed simultaneously would it be possible to treat one clock as ground truth while using the other for \systemname{}. 

As a concrete example, consider two clients, A and B, each reading two clocks in sequence. Client A first reads \texttt{clock1} at time $t_1$ and then reads \texttt{clock2} at time $t_4$. Client B reads \texttt{clock1} at time $t_2$ followed by \texttt{clock2} at time $t_3$. Assume that \texttt{clock1} timestamps the event occurrence, while \texttt{clock2} is a perfectly synchronized reference clock, and that the observed times satisfy $t_1 < t_2 < t_3 < t_4$.

Under \texttt{clock1}, the event at client A precedes the event at client B. However, under the perfectly synchronized \texttt{clock2}, the event at client B precedes the event at client A. This discrepancy arises from asymmetric clock-access latencies, showing how even a perfectly synchronized reference clock may fail to provide ground truth in practice. This is problematic for \systemname{} evaluation as most of its benefit becomes apparent when inter-event-occurrence duration is small e.g., a few microseconds. But a few microseconds latency of variance in accessing two clocks can render the ground truth useless. Consequently, establishing reliable ground truth for event ordering remains non-trivial, which motivates our design of simulation-based methodology described below.

\Para{Simulator-based setup:} We use the ns-3 simulator to instantiate multiple clients and a server acting as the sequencer. Event generation at clients is explicitly controlled, allowing us to prescribe a ground-truth order of events. For example, an input sequence $c^1_i, c^2_j, c^3_k$ specifies that client $i$ generates the first event, followed by client $j$, and then client $k$. The inter-event durations are configurable.

When an event occurs at a client, the client records a local timestamp using a simulated virtual clock with configurable offset and drift. These timestamps are provided as inputs to the sequencer, which produces an ordered sequence of events. We evaluate \systemname{} by comparing the sequencer’s output order against the prescribed input sequence.

Each client and the server has a simulated virtual clock. The offsets and drifts for each clock are selected from a normal distribution whose parameters are configurable. Clients run NTP protocol to synchronize their clocks with the server's clock which is assumed to be the global time. The clock sync. accuracy depends on the latency experienced by the packets containing the synchronization probes. To make the synchronization realistic, we configure the latency distribution to what is expected in a real data center. 

\Para{Realistic latency distribution:} In a separate ns-3 simulation, we simulate a fat-tree data center topology from Homa~\cite{homa, homa_ns3}: 144 hosts, 9 ToRs, 4 spines where each ToR downlink is connected to mutually exclusive set of 16 hosts and each ToR's uplink is connected all spines. ToR downlinks' have a capacity of 10Gbps each where the uplinks' have a capacity of 40 Gbps each. More details about the workload in Appendix~\ref{app:dcn_setup}. 
% We set the workload ``DCTCP workload'' from~\cite{homa, homa_ns3}. 
We gather latencies experienced by the packets in this simulation. These latencies are used in our \systemname{} simulation so that clock sync. probes can experience near-realistic conditions. 
% In a separate ns-3 simulation (Appendix~\ref{app:datacenter-simulation}), we simulate a data-center workload to get realistic latencies which we feed to \systemname{} simulation. 

\Para{Baseline:} We develop a baseline to compare against \systemname{}. It is inspired by Spanner's TrueTime API~\cite{spanner}. Each event timestamp is represented with an uncertainty interval around it i.e., for a timestamp $t$, we represent it as $[t - Error Bound,\, t + Error Bound]$. In the paper~\cite{spanner}, the $Error Bound$ is set to be a limit that the clock error generally will not cross and is highly dependent on the production environment. We infer the bound from the gathered clock sync. probes. We find the standard deviation, $\sigma$, from the gathered probes and set the bound to be $3\sigma$. Once all event timestamps are represented as uncertainty intervals, we sort the intervals based on the starting point of intervals. The sorted intervals directly represent the order of the corresponding events while assigning the same order/rank to events if their intervals overlap. The above mechanism helps finding the order of a given set of events, while the online extension uses the same protocol presented in Algorithm~\ref{alg:online-stable-time}. 

\Para{Configuration:} We use the following configuration unless specified otherwise. We have 100 clients with IDs: $c_1, c_2, ..., c_{100}$, generating events in the order of their client ID. Total 200 events are generated per experiment, while we perform 5 experiments and average out the fairness. The clock sync probes' frequency is set to 10 probes per second. The offset for each client's clock is taken from  normal distribution with mean of 0 and standard deviation of \SI{1}{\second}. 

\Para{Evaluation questions:} We focus on the following aspects: 
\begin{enumerate}[labelwidth=!,nosep,label=(\arabic*),leftmargin=*]

\item How does fairness achieved by \systemname{} compare to that achieved by TrueTime-baseline?
\item How do the different characteristics of an application (e.g., inter-messages delays, number of messages/clients etc.) impact the fairness achieved by \systemname{}?
\item How is fairness impacted by the network latencies?
\item How does clock sync protocol affect the fairness?
\item Dive into several components of \systemname{}.
% \item Dive into evaluation setup. \todo{if time and space allows for it}
% \item Comparison against fairness mechanisms employed by financial exchanges.
% \item Other than the above, we accompany the evaluation with various sanity check of the simulation and deep dive into the setup used for evaluation?

\end{enumerate}

\Para{Comparison Metric: Rank Agreement Score (RAS)} is proposed to quantify fairness given a ground truth ordered sequence and an ordered sequence output by \systemname{} or the baseline. RAS is computed as follows: for each correctly ordered pair of events, score increases by 1, for each incorrectly ordered pair, the score decreases by 1 while a pair of events left unordered w.r.t. each other adds 0 to the score. The resultant score is normalized by the total number of pairs of events. While this metric is helpful in most scenarios, it has some pitfalls we elaborate further while explaining the evaluation of \systemname{}. Higher RAS is better. 

% \begin{figure}
%     \centering
%     \includegraphics[width=1\linewidth]{figures/fairness_plot_sweep_intermessagesdelay.pdf}
%     \caption{\systemname{} usually achieves better fairness than the TrueTime baseline.}
%     \label{fig:fairness_sweep_delay}
% \end{figure}

\begin{figure}[!t]
    \centering

    \hfill
    \begin{minipage}[t]{0.48\linewidth}
        \centering
        \includegraphics[width=1\linewidth]{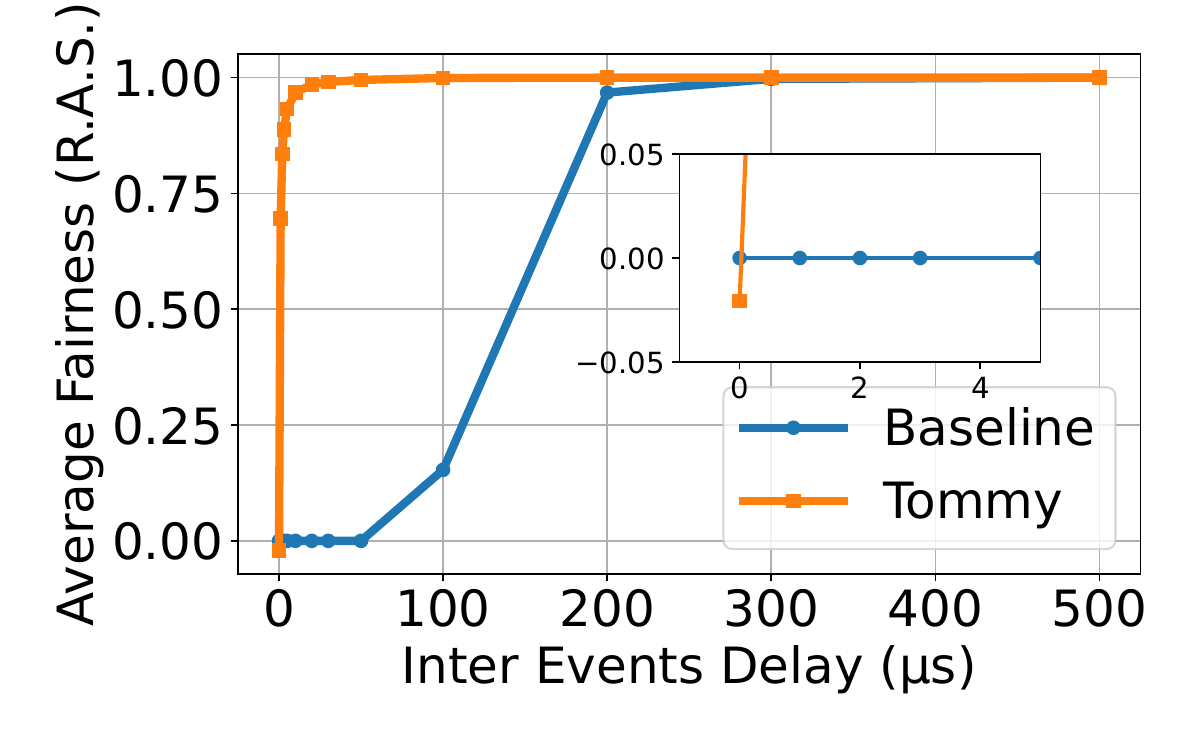}
        % \vspace{-0.7cm}
        \caption{\textmd{\systemname{} mostly achieves better fairness than the baseline.}}
        \label{fig:fairness_sweep_delay}
        \end{minipage}
    \hfill
    \begin{minipage}[t]{0.48\linewidth}
        \centering
        \includegraphics[width=1\linewidth]{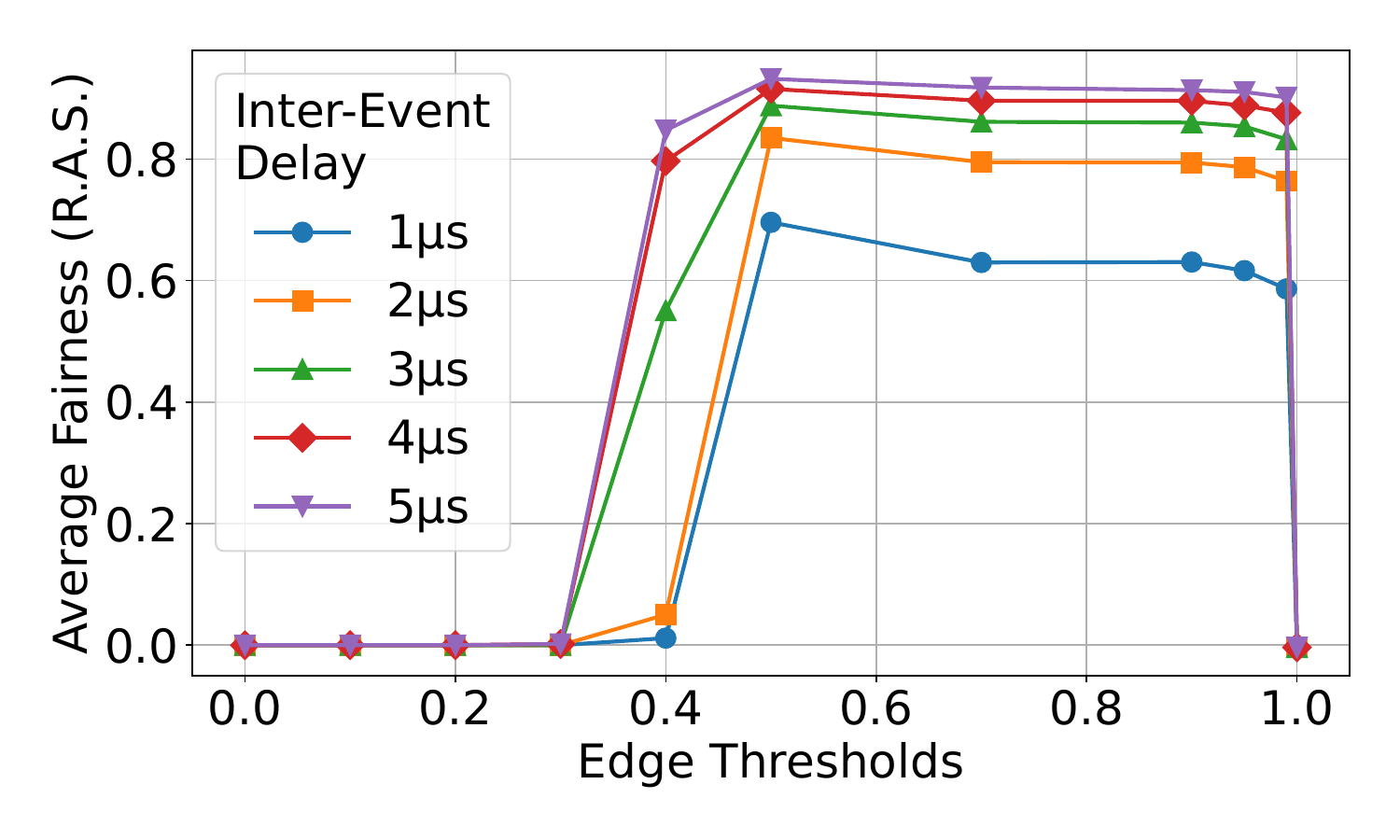}
        % \vspace{-0.7cm}
        \caption{\textmd{Optimal point for fairness arrives at edge threshold of 0.5.}}
        \label{fig:fairness_sweep_messages_edge_thresh}
    \end{minipage}
    \vspace{-0.7cm}
\end{figure}

\begin{figure*}
    \centering

    \begin{minipage}[t]{0.32\linewidth}
        \centering
        \includegraphics[width=\linewidth]{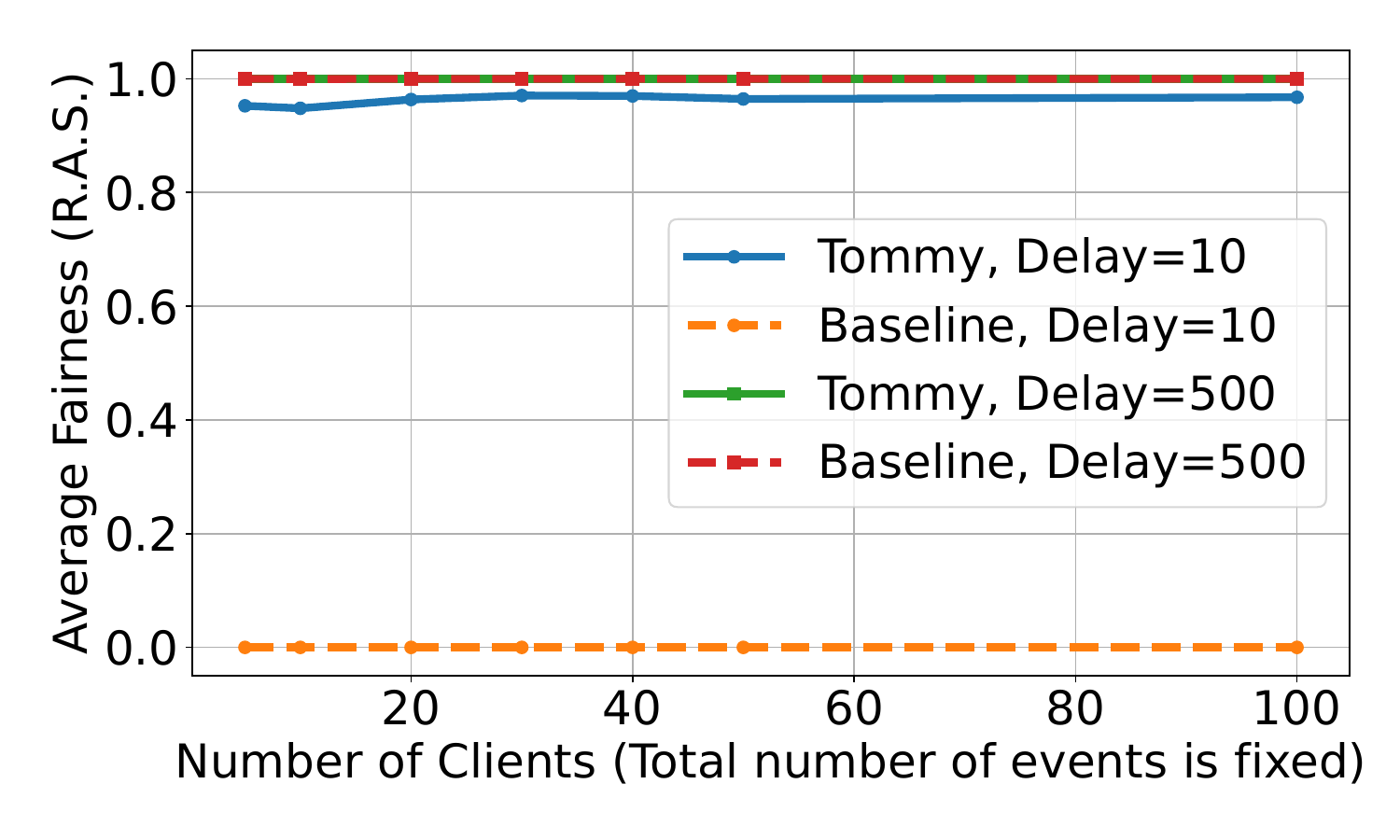}
        \vspace{-0.7cm}
        \captionof{figure}{\textmd{Number of clients (with no change in total events) have no effect on fairness.}}
        \label{fig:fairness_sweep_clients}
    \end{minipage}
    \hfill
    \begin{minipage}[t]{0.32\linewidth}
        \centering
        \includegraphics[width=\linewidth]{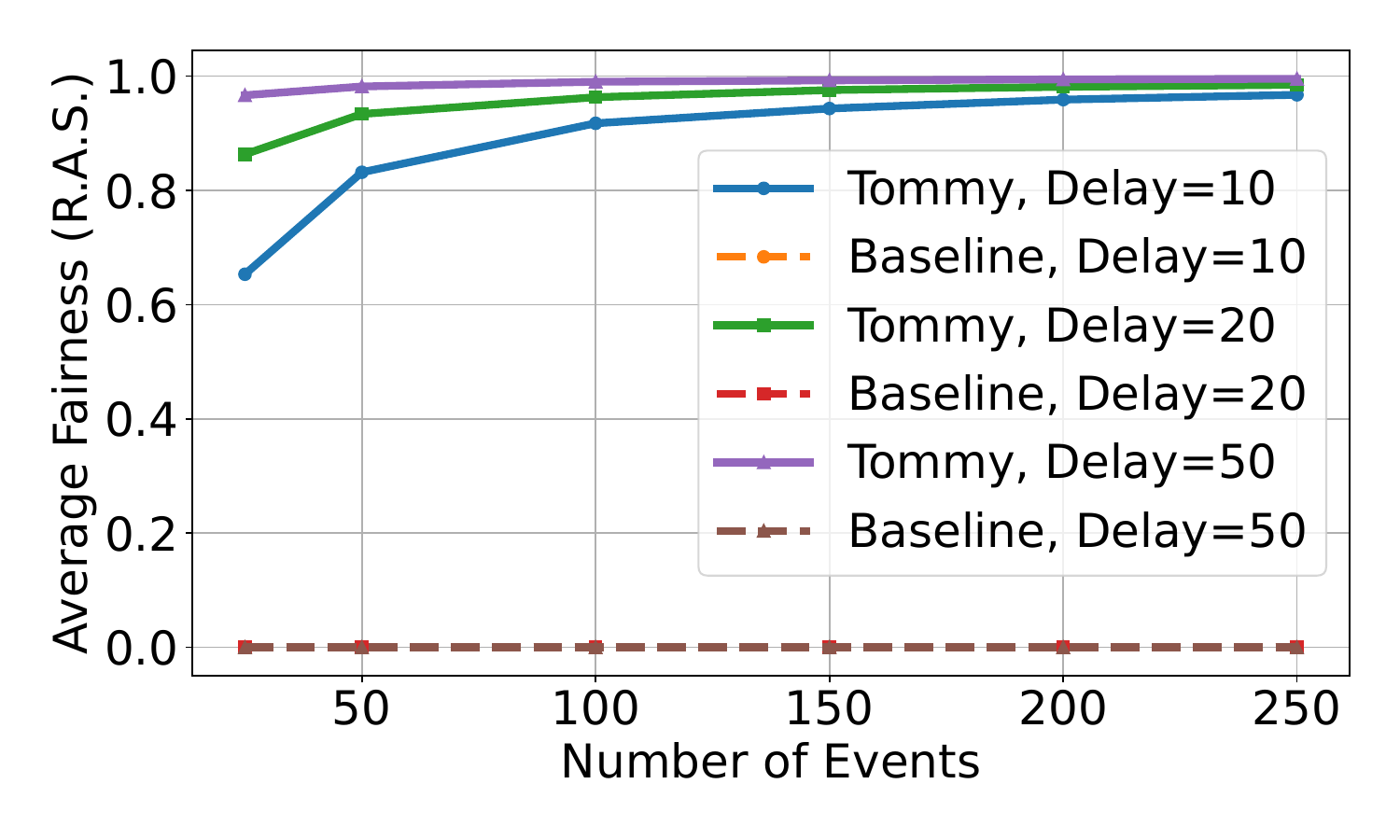}
        \vspace{-0.7cm}
        \captionof{figure}{\textmd{Fairness with \systemname{} is higher if ordering more events, an artifact of the fairness metric.}}
        \label{fig:fairness_sweep_messages}
    \end{minipage}
    \hfill
    \begin{minipage}[t]{0.32\linewidth}
        \centering
        \includegraphics[width=\linewidth]{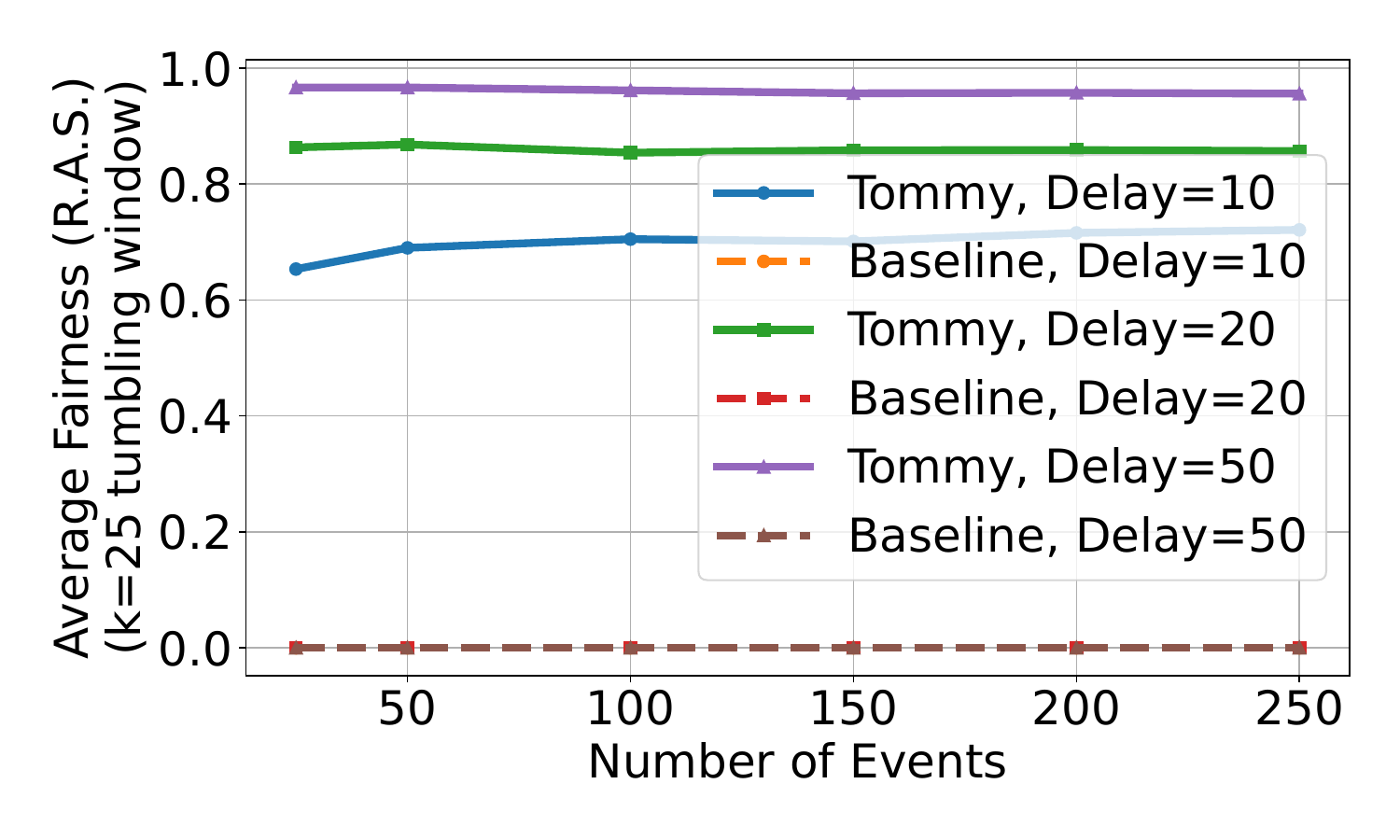}
        \vspace{-0.7cm}
        \captionof{figure}{\textmd{Fairness (over a constant sized window of events) stays constant as total number of events increases.}}
        \label{fig:fairness_sweep_messages_constant_step}
    \end{minipage}
\vspace{-0.2cm}
\end{figure*}

\subsection{\systemname{} vs. TrueTime Baseline}
\label{subsec:tommy_vs_baseline}

To show effectiveness of \systemname{} in fairly ordering events, we compare it against the baseline discussed above. 

Figure~\ref{fig:fairness_sweep_delay} shows \systemname{} mostly achieves higher fairness compared to the baseline as we sweep the inter-event delay from 0 to \SI{500}{\micro\second}. At higher delays, ordering the events is easier as the uncertainty by clocks is smaller compared to the gap between event timestamps. Both \systemname{} and the TrueTime baseline achieves perfect fairness with delay=\SI{500}{\micro\second}. Towards lower delays, \systemname{} establishes its effectiveness over the baseline. However, the baseline is a deterministic mechanism and may leave events unordered if it is uncertain achieving fairness closer or equal to 0. On the other hand, \systemname{}, as a probabilistic mechanism, may attempt to order events even under high uncertainty and make mistakes, leading to negative fairness in some scenarios; baseline's fairness never goes below 0. That's the reason for baseline achieving higher fairness than \systemname{} in Figure~\ref{fig:fairness_sweep_delay} with inter-event delay equal to 0. As clock's errors increase and the inter-event delay decreases, \systemname{} may make more ordering mistakes than the baseline, but such scenarios would be out of the operational regions for the most systems. We do one experiment with pathological configuration i.e., setting clock drift rate to be extremely high (100 ppm) and see that \systemname{} can achieve $\approx 18\%$ lower fairness (RAS) than the baseline. 

% In the subsequent sections, it is reinforced that mostly \systemname{} performs better than the baseline. 

\subsection{Properties Of An Application And The Impact on Fairness}

Several properties of an application e.g., how many clients are involved, how many events occur per second, the inter-event delay can have impact on the fairness achieved by \systemname{}. We already show the impact of inter-event delay on the fairness in \S\ref{subsec:tommy_vs_baseline}. Here, we discuss the impact of number of clients and events, as well as a pitfall of the RAS metric. 

\Para{Changing the number of clients:} For a fixed number of events, evenly distributed among clients, changing the number of clients has no impact on RAS in our simulations. Figure~\ref{fig:fairness_sweep_clients} shows how the RAS remains almost constant as we sweep the number of clients. The effect stays consistent regardless of the inter-event delays (\SI{10}{\micro\second} and \SI{500}{\micro\second} are shown in the figure). In these experiments, the extra knowledge, that the events occurring at the same client can be ordered using the happened-before relation, is not utilized as otherwise the ordering may improve as the number of clients is made small and the number of events stays constant. 

\Para{Changing the number of events:} We fix the number of clients to 25 and sweep the number of events from 25 to 250. Figure~\ref{fig:fairness_sweep_messages} shows that \systemname{} fairness improves as the number of events increase. This is a surprising result and upon further investigation, we realize it is an artifact of our RAS metric. As the number of events increase, a later event is ordered correctly with a large number of older events as the inter-event delay becomes high. Due to this, the number of correctly ordered pairs increases, so RAS increases. To isolate fairness from this effect, we only measure fairness for events that occur close together, i.e., we use a window (=25 events) and compute RAS for events only in that window. Figure~\ref{fig:fairness_sweep_messages_constant_step} shows that increasing the number of events has no significant impact on the fairness computed over a constant sized window, which is an expected behavior.  

% Fairness is not directly related to the number of clients/events, but it depends on the inter-event delay as shown previously. If changing the clients/events leads to the changes of the inter-event delay, then the impact on fairness would show up. However, while running the experiments above,  we have controlled for the inter-event delay regardless of the number of clients/events.  

Fairness is not directly related to the number of clients or events; rather, it depends on the inter-event delay, as shown previously. If changing the number of clients/events alters inter-event delay, the resulting impact on fairness will manifest. However, in the experiments above, we control for the inter-event delay, regardless of the number of clients/events.

\subsection{Network Latencies Impact on Fairness}
\label{subsec:evals_network_lats}

\begin{figure}[!t]
    \centering

    \begin{minipage}[t]{0.48\linewidth}
        \centering
        \includegraphics[width=\linewidth]{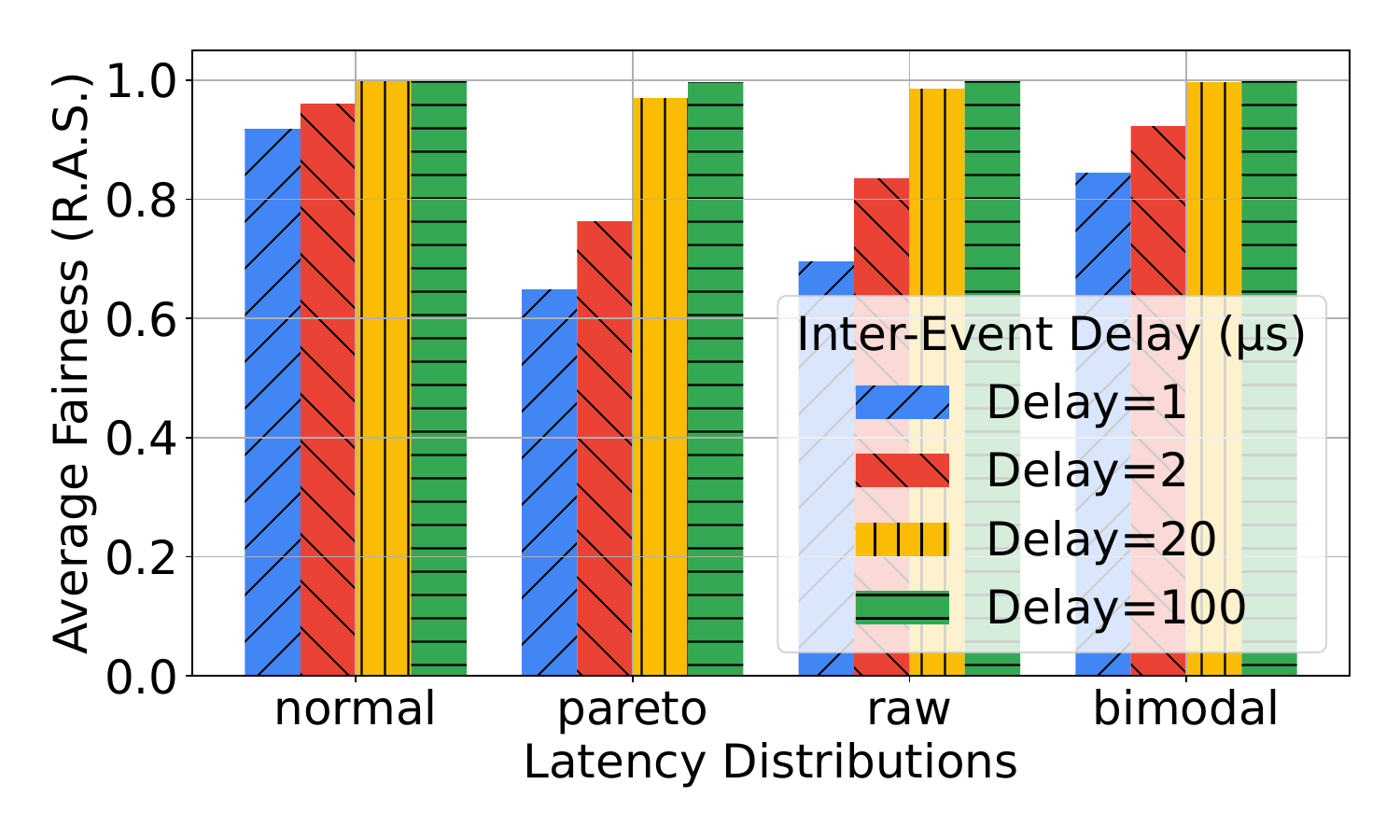}
        \captionof{figure}{\textmd{Latency models  impact fairness.}}
        \label{fig:fairness_sweep_latency_models}
    \end{minipage}
    \hfill
    \begin{minipage}[t]{0.48\linewidth}
        \centering
        \includegraphics[width=\linewidth]{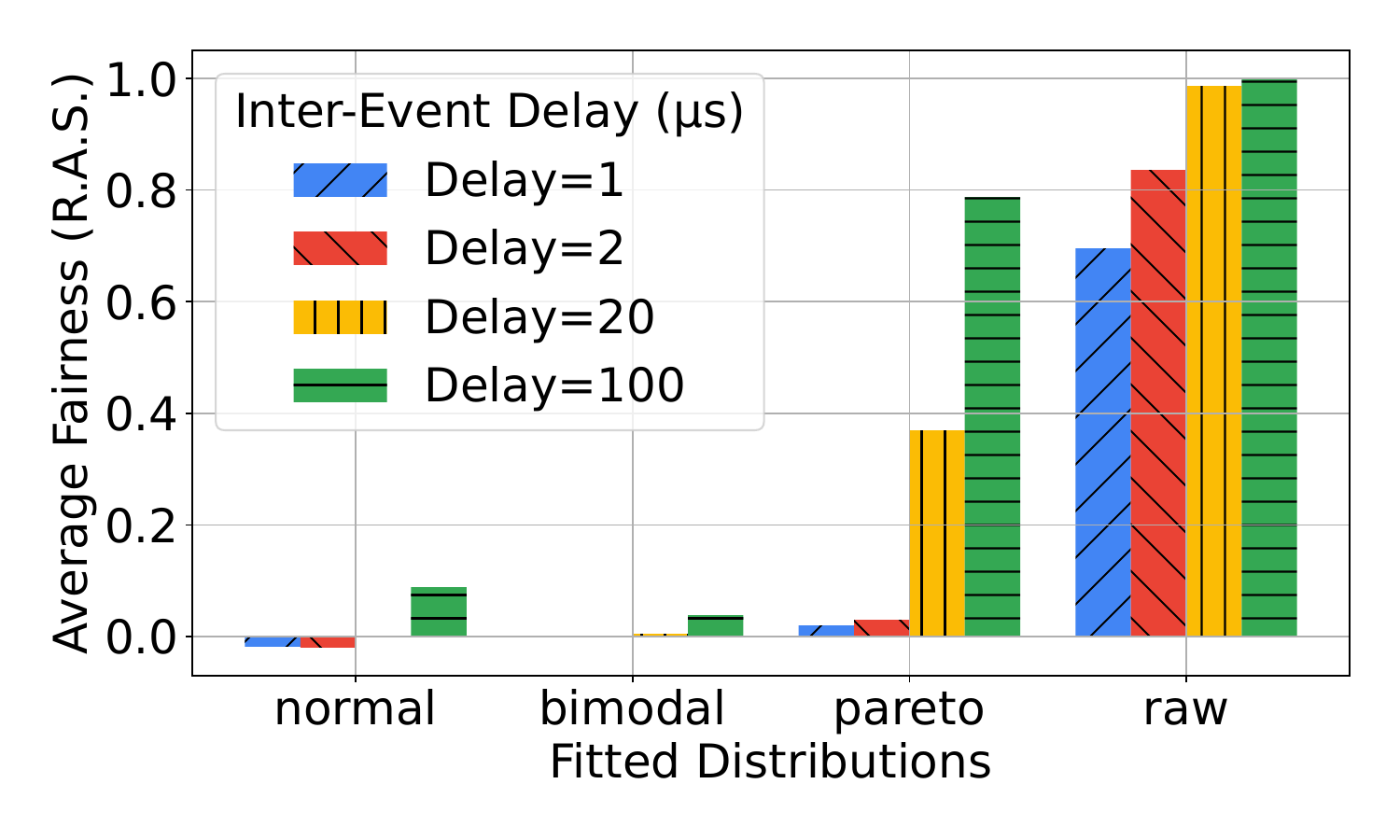}
        \captionof{figure}{\textmd{Fitting distributions not better.}}
        \label{fig:fairness_sweep_fit_dists}
    \end{minipage}
\vspace{-0.8cm}
\end{figure}

Network latencies can impact the clock synchronization probes as well as the messages that convey the event timestamps to the sequencer. Sequencer ensures that latencies of messages carrying the event timestamps have no effect on fairness as the sequencer only progresses once it has seen certain messages from \emph{every} client so the fairness calculations are idempotent w.r.t how much time certain messages took to arrive. 
% In the worst case, sequencer will halt and not make progress if certain messages take too long to arrive, but that is a concern we do not indulge with and should be explored further during a production realization of \systemname{}. 

The impact of network latencies on clock synchronization probes, however, is noticeable as it dictates the clock synchronization accuracy and the clock correction distributions. First, for a sanity check, we feed several random latency distributions to our ns-3 \systemname{} simulation and study the impact on fairness. Figure~\ref{fig:fairness_sweep_latency_models} shows the fairness varies across multiple latency distributions. Second, as we collect latencies of packets from the fat-tree data center simulation, we fit a distribution to that data and then feed the learnt distribution to our \systemname{} simulation. Figure~\ref{fig:fairness_sweep_fit_dists} presents the fairness achieved by not fitting any distribution and using the obtained histogram of latencies as-is achieves the highest fairness compared to fitting a normal, bimodal or even a pareto distribution. This establishes the sensitivity of \systemname{} to the latencies experienced by synchronization probes. 

\subsection{Effect of Synchronization Probes' Frequency}
\label{subsec:evals_probes_freq}

\begin{figure}[t]
    \centering

    \begin{minipage}[t]{0.48\linewidth}
        \centering
        \includegraphics[width=\linewidth]{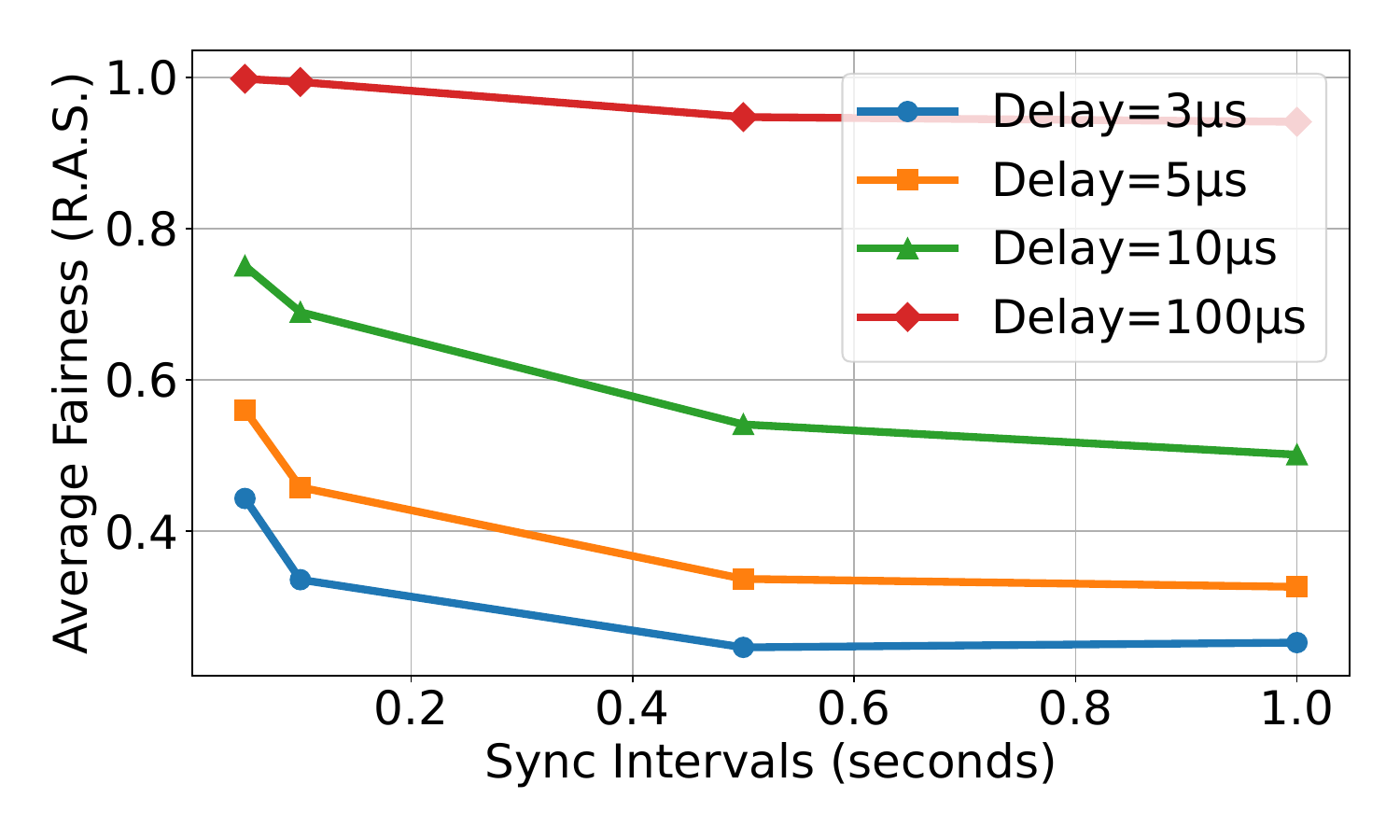}
        \vspace{-0.5cm}
        \captionof{figure}{\textmd{Clock sync. interval impact fairness.}}
        \label{fig:fairness_sweep_syncintervals}
        \vspace{-0.5cm}
    \end{minipage}
    \hfill
    \begin{minipage}[t]{0.48\linewidth}
        \centering
        \includegraphics[width=\linewidth]{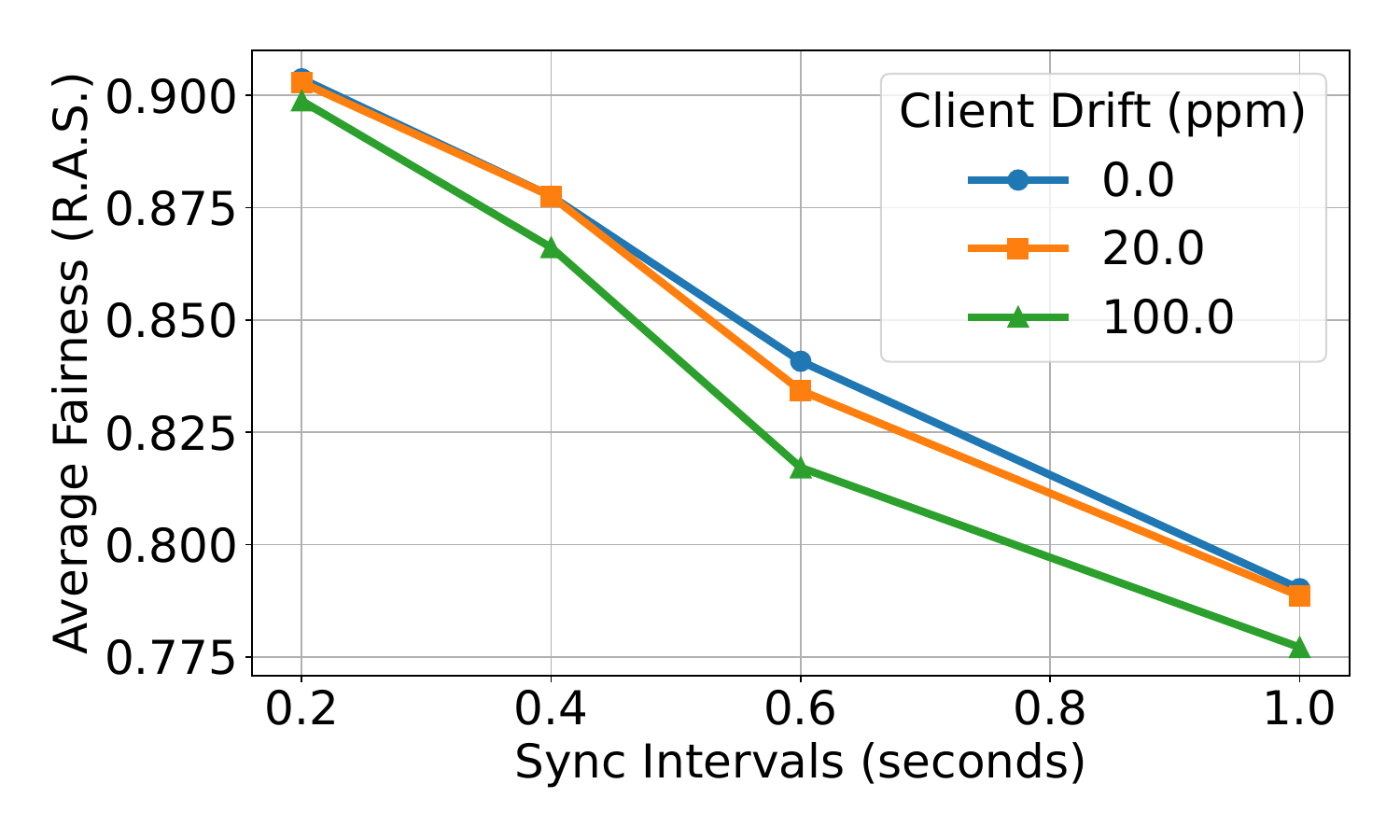}
        \vspace{-0.5cm}
        \captionof{figure}{\textmd{Higher drift lowers the fairness.}}
        
        \label{fig:fairness_sweep_syncintervals_drifts}
        \vspace{-0.5cm}
    \end{minipage}

% \vspace{-0.5cm}
\end{figure}

Clock synchronization probes' frequency directly impacts the fidelity of clock correction distributions estimated by clients. These distributions in turn directly impact the fairness achieved by \systemname{}. Figure~\ref{fig:fairness_sweep_syncintervals} shows the length of the clock sync. interval, in our implementation of NTP protocol which synchronizes the virtual clock in our ns-3 simulation, is inversely proportional to the fairness achieved by \systemname{}. 

As one of our assumption of a deployed clock synchronization protocol is that it is adequately addressing the clock drift. Figure~\ref{fig:fairness_sweep_syncintervals_drifts} uses the inter-event delay of \SI{15}{\micro\second}, and it shows how higher drift overall leads to lower fairness, while the synchronization intervals' effect also persists. 
%This experiment utilized inter-event delay of \SI{15}{\micro\second}. 

\subsection{Dive into \systemname{} components}
\label{seubsec:evals_tommy_algos}

\Para{Effect of edge threshold:} In \S\ref{sec:graph}, we only consider an edge from $i$ to $j$ if $i \xrightarrow{p > 0.5}j $. The threshold 0.5 was chosen because it indicates that it is more likely that $i$ happened before $j$ than $j$ before $i$. However, increasing the value of this threshold may provide more confidence in the ordering as we only consider edges which are much more likely to be true. But on the other hand, setting this threshold high bars several pairwise relations providing less freedom for \systemname{} to order the events. Figure~\ref{fig:fairness_sweep_messages_edge_thresh} shows that as the edge threshold increases beyond 0.5, the fairness does not increase, but a slight decrease occurs. Further, setting threshold to lower than 0.5, sharply decreases the achieved fairness as \systemname{} stays indifferent about several events.

% \begin{figure}
%     \centering
%     \includegraphics[width=1\linewidth]{figures/fairness_plot_sweep_edge_thresh.pdf}
%     \caption{Optimal point for fairness arrives at edge threshold of 0.5.}
%     \label{fig:fairness_sweep_messages_edge_thresh}
% \end{figure}

\Para{Stability Condition:} As clock errors increase, the correction distributions widen. Larger distributions in turn lead to larger future timestamps that needs to be considered with a current event for it to be considered stable in an online setting. 
In Figure~\ref{fig:fairness_sweep_stable}, we perform an experiment where we increase the clock error (drift) and study the increase in the future timestamps calculated for stability condition (Algorithm~\ref{alg:online-stable-time}), compared to the case with zero clock drift. Figure~\ref{fig:fairness_sweep_stable} shows that the clock error and the increment in stable timestamps are directly proportional. 

% \begin{figure}
%     \centering
%     \includegraphics[width=1\linewidth]{figures/fairness_plot_sweep_stable.pdf}
%     \caption{Future stable timestamps become larger as error increases}
%     \label{fig:fairness_sweep_stable}
% \end{figure}

\begin{figure}[!t]
    \centering

    \hfill
    \begin{minipage}[t]{0.48\linewidth}
        \centering
        \includegraphics[width=\linewidth]{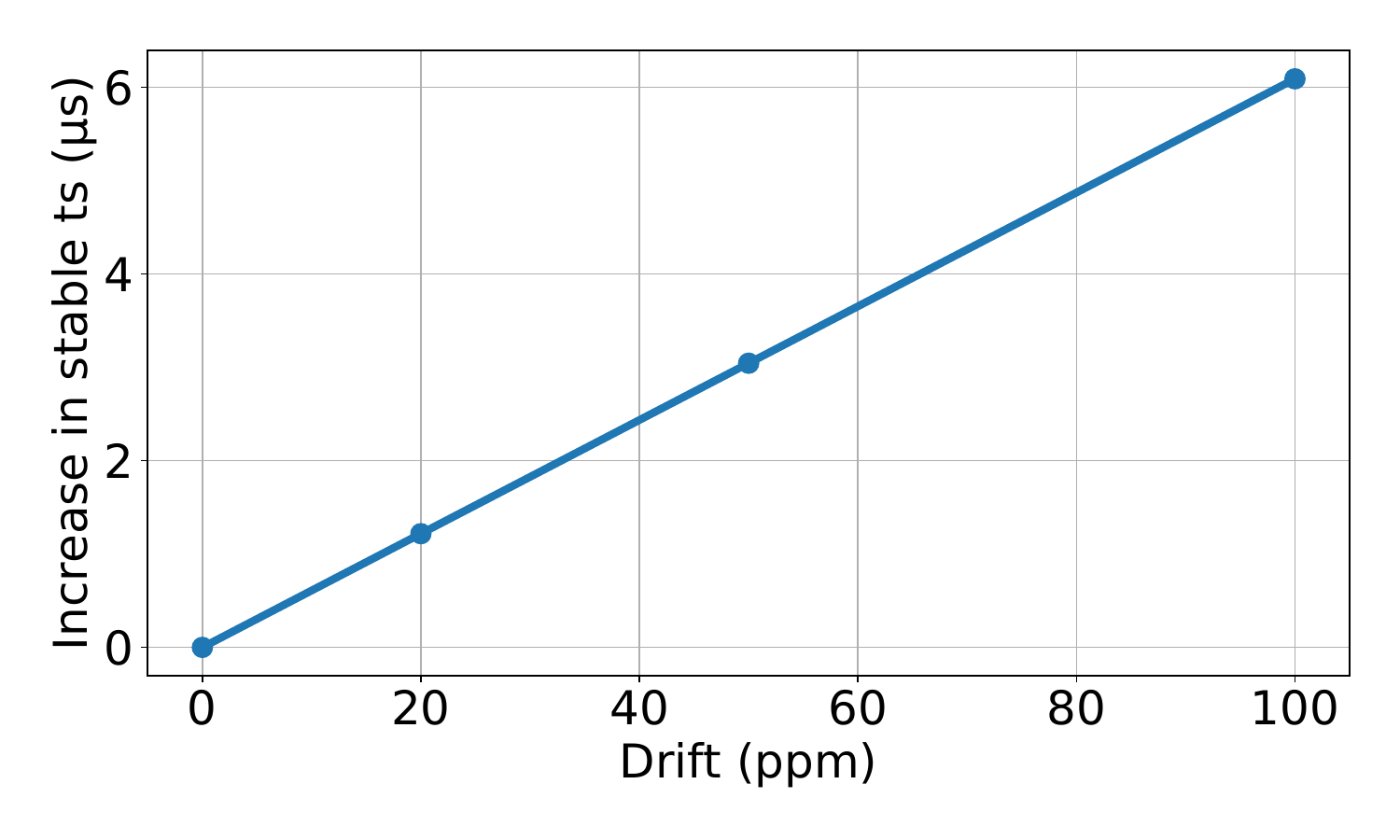}
        \vspace{-0.7cm}
        \captionof{figure}{\textmd{Future stable ts enlarge error grows.}}
        % \vspace{-0.5cm}
        \label{fig:fairness_sweep_stable}
    \end{minipage}
    \hfill
    \begin{minipage}[t]{0.48\linewidth}
    \centering
    \includegraphics[width=1\linewidth]{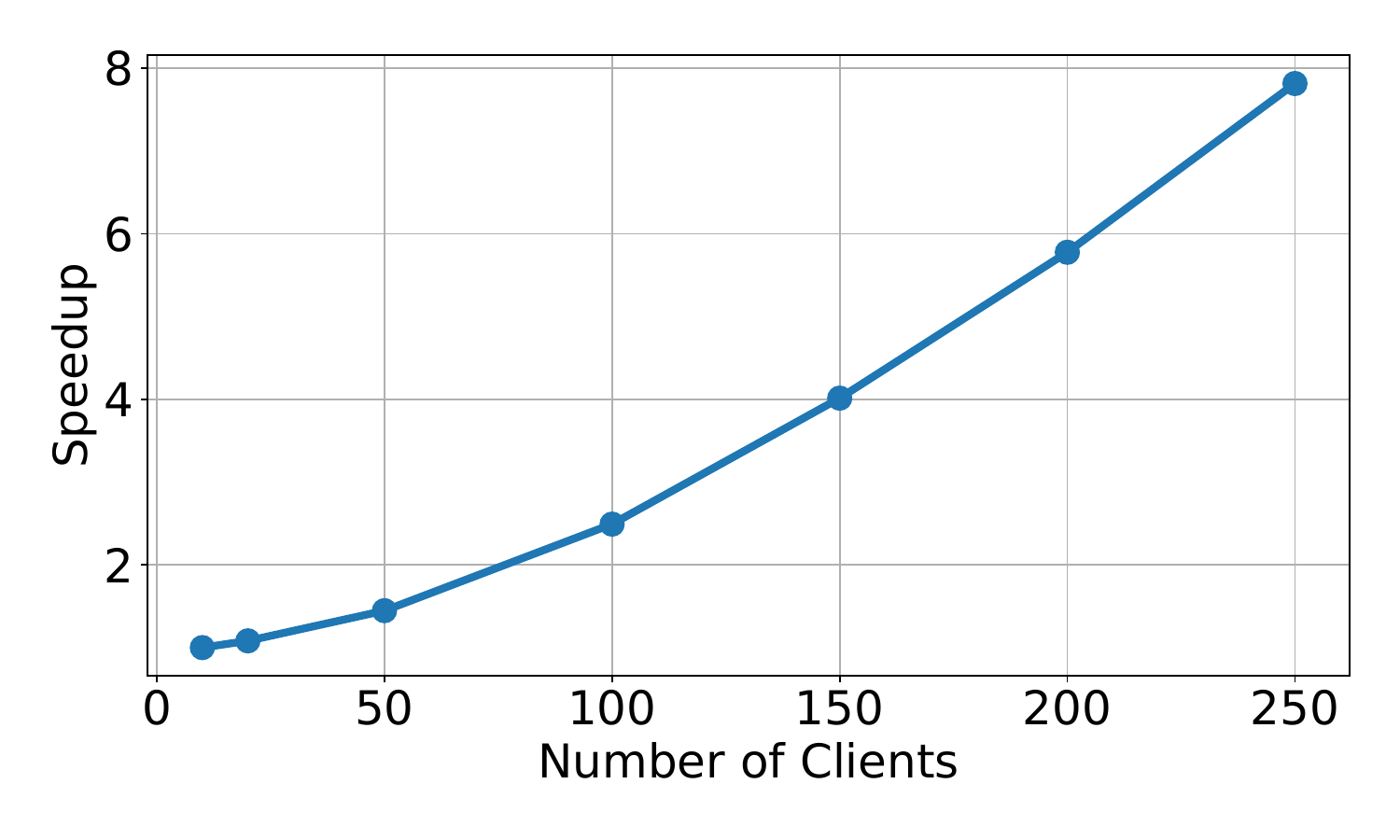}
    \vspace{-0.7cm}
    \caption{\textmd{Pre-computing achieves speed up.}}
    % \vspace{-0.5cm}
    \label{fig:fairness_sweep_precompute}
    \end{minipage}
    \vspace{-0.7cm}

\end{figure}

\Para{Benefit of pre-computing difference distributions:} A main optimization is to pre-compute corrections-difference distributions for each pair of clients as it significantly improves the time it takes finding the pairwise preceding probability for events. Figure~\ref{fig:fairness_sweep_precompute} shows as the number of clients increases, the speedup (=time taken without optimization / time taken with optimization) increases. In this experiment, each client generates 10 events in the same round robin fashion as described in the evaluation setup previously. 

% \begin{figure}
%     \centering
%     \includegraphics[width=1\linewidth]{figures/fairness_plot_sweep_precompute.pdf}
%     \caption{Pre-computing difference distributions achieves significant speed up.}
%     \label{fig:fairness_sweep_precompute}
% \end{figure}

% i will include it later 
% \input{application-hedging}

%% file: limitations.tex
\section{Limitations And Future Work}
\label{sec:limitations}

Tommy has limitations stemming from modeling assumptions and evaluation scope.

\Para{Event definition:} We treat an event as occurring when a client reads its local timestamp. In practice, application level actions may precede this read by variable delays, introducing additional noise that is not captured by our model.

\Para{Modeling assumptions:} Our model assumes independence and short term stability of clock correction distributions. Correlated errors across clients or abrupt network changes may violate these assumptions and degrade ordering accuracy. Addressing such effects requires more complex models and adaptive mechanisms.

\Para{Evaluation scope:} Our results are based on ns-3 simulations. While this enables controlled experiments and access to ground truth, real world deployments may exhibit additional sources of variability that we leave to future work.

%% file: future-work.tex
\Para{Future Work: } There are several promising directions for future work building on Tommy: (i) Strengthening the failure model. Tommy currently assumes fail stop clients. Extending the design to tolerate Byzantine clients would broaden its applicability, especially in open or adversarial environments. (ii) Scaling the sequencer. While this paper focuses on a single logical sequencer, scaling fair ordering to larger number of clients and larger workloads is a natural extension. (iii) A real-world deployment would allow deeper study of clock correction stability, correlation effects, and end to end application level fairness. It would also enable exploration of adaptive mechanisms that dynamically adjust sequencing behavior in response to changing network conditions.

%% file: conclusion.tex
\section{Conclusion}
\label{sec:conclusion}

We present \systemname{}, a sequencer that achieves probabilistic fair ordering of events in the presence of inevitable clock synchronization errors. Instead of attempting to eliminate timestamp uncertainty, \systemname{} embraces it by modeling clock corrections and using probabilistic comparisons to order events fairly. By reducing fair ordering to a ranking problem with intransitive preferences, \systemname{} leverages social choice theory to produce meaningful partial ordering of events. Our evaluation demonstrates that Tommy achieves better fairness than existing baselines, highlighting the feasibility of fair sequencing without relying on tightly synchronized clocks. This work does not raise any ethical concerns.

%% file: appendix.tex
\newpage
\appendix
\section{Transivity holds for Gaussian Distributions}
\label{app:guassian}
\begin{proposition} Let $X,Y,Z$ be independent normal random variables
\[
X\sim\mathcal N(\mu_X,\sigma_X^{2}),\qquad
Y\sim\mathcal N(\mu_Y,\sigma_Y^{2}),\qquad
Z\sim\mathcal N(\mu_Z,\sigma_Z^{2}).
\]
Define the preference relation
\[
X \succ Y \;\Longleftrightarrow\; \Pr[X>Y]>\tfrac12.
\]
Then $\succ$ is transitive: if $X\succ Y$ and $Y\succ Z$, we necessarily have $X\succ Z$.
\end{proposition}

\begin{proof}
For any two independent Gaussian variables $A\sim\mathcal N(\mu_A,\sigma_A^{2})$ and
$B\sim\mathcal N(\mu_B,\sigma_B^{2})$, the difference $A-B$ is Gaussian with
\[
A-B\;\sim\;\mathcal N\!\bigl(\mu_A-\mu_B,\;\sigma_A^{2}+\sigma_B^{2}\bigr).
\]
Hence
\[
\Pr[A>B]=\Pr[A-B>0]
        =\Phi\!\Bigl(\frac{\mu_A-\mu_B}{\sqrt{\sigma_A^{2}+\sigma_B^{2}}}\Bigr),
\]
where $\Phi$ is the standard–normal CDF. Now:

\begin{equation}\label{eq:start}
    \Pr[A>B]>\tfrac12
    \quad\Longleftrightarrow\quad
    \Phi\!\Bigl(\frac{\mu_A-\mu_B}{\sqrt{\sigma_A^{2}+\sigma_B^{2}}}\Bigr) >\tfrac12.
\end{equation}
% \]

As \(\Phi(0) = \tfrac12\), so:

\[
    \Phi\!\Bigl(\frac{\mu_A-\mu_B}{\sqrt{\sigma_A^{2}+\sigma_B^{2}}}\Bigr) >\Phi(0).
\]

Because $\Phi$ is a strictly increasing function,
\[
\Phi\left( \frac{\mu_A - \mu_B}{\sqrt{\sigma_A^2 + \sigma_B^2}} \right) > \Phi(0)
\quad \Longleftrightarrow \quad
\frac{\mu_A - \mu_B}{\sqrt{\sigma_A^2 + \sigma_B^2}} > 0.
\]

As the denominator $\sqrt{\sigma_A^2 + \sigma_B^2}$ cannot be negative, 

\begin{equation}
    \label{eq:means_matter}
    \Pr[A>B]>\tfrac12
    \quad\Longleftrightarrow\quad
    \mu_A-\mu_B>0
    \quad\Longleftrightarrow\quad
    \mu_A>\mu_B.
\end{equation}

Thus our preference rule depends \emph{only} on the means.

Now, suppose $X\succ Y$ and $Y\succ Z$.  This implies
\[
\mu_X>\mu_Y \quad\text{and}\quad \mu_Y>\mu_Z,
\]
which together give $\mu_X>\mu_Z$ because means (i.e., real numbers) are transitive. Applying eq. \ref{eq:means_matter} to $\mu_X>\mu_Z$, yields
$X\succ Z$.
\end{proof}

\section{Data-center workload setup}
\label{app:dcn_setup}

Each client has a 10Gbps uplink to a ToR. Clients are seeded with a message-size distribution from Homa~\cite{homa}, also shown in Figure~\ref{fig:dctcp-workload}. Clients can be configured to emit a load of certain ratio for example 0.1 denotes 10\% of the uplink capacity is exhausted. TCP Cubic is used. Figure~\ref{fig:load_vs_latency} shows the latencies of packets in this data-center simulation as we configure clients to several load ratios. The simulation ran for 0.05s of data-center workload duration (the time taken by the simulator was much longer) for each load ratio. Figure~\ref{fig:dcn_latencies} shows the histogram of latencies with load ratio of 1.0, showing multiple modes and a long tail. When running \systemname{} simulation with this latency data, the clock correction distributions obtained are shown in Figure~\ref{fig:correction_dists} (for two clients). These correction distributions also exhibit multiple modes. 
% In a separate ns-3 simulation, we simulate a fat-tree data center topology: 144 hosts, 9 ToRs, 4 spines where each ToR downlink is connected to mutually exclusive set of 16 hosts and each each ToR's uplink is connected all spines. ToR downlinks' have a capacity of 10Gbps each where the uplinks' have a capacity of 40 Gbps each. We set the workload ``DCTCP workload'' from~\cite{homa, homa_ns3}. The above is the same setup used by ns-3 simulations of Homa~\cite{homa_ns3}. As we simulate the above workload, we gather latencies experienced by the packets. These latencies are used in our \systemname{} simulation so that clock sync. probes can experience realistic conditions. 

\begin{figure}
    \centering
    \includegraphics[width=0.75\linewidth]{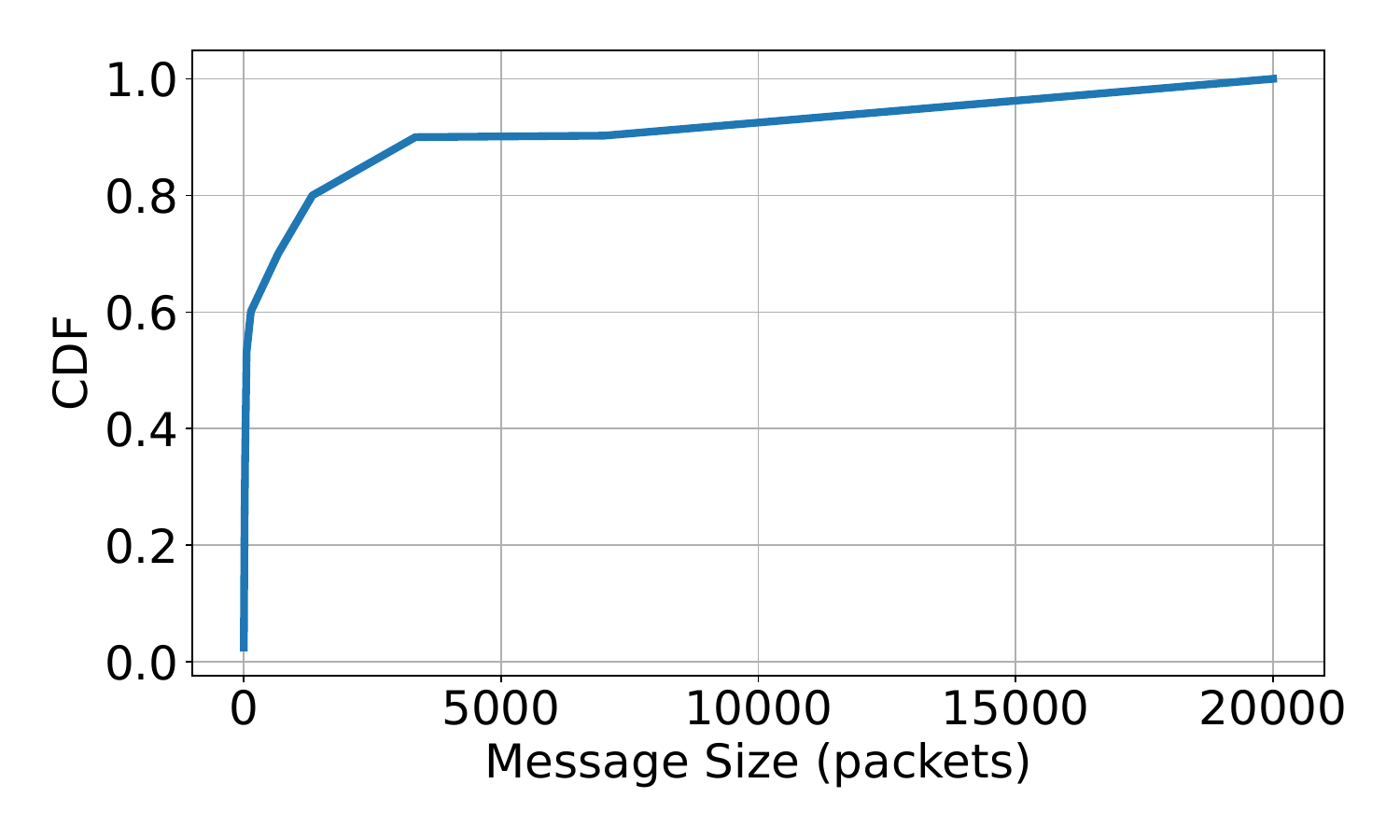}
    \captionof{figure}{\textmd{DCN workload from Homa~\cite{homa, homa_ns3}}}
    \label{fig:dctcp-workload}
\end{figure}
 
\begin{figure}
        \centering
        \includegraphics[width=0.75\linewidth]{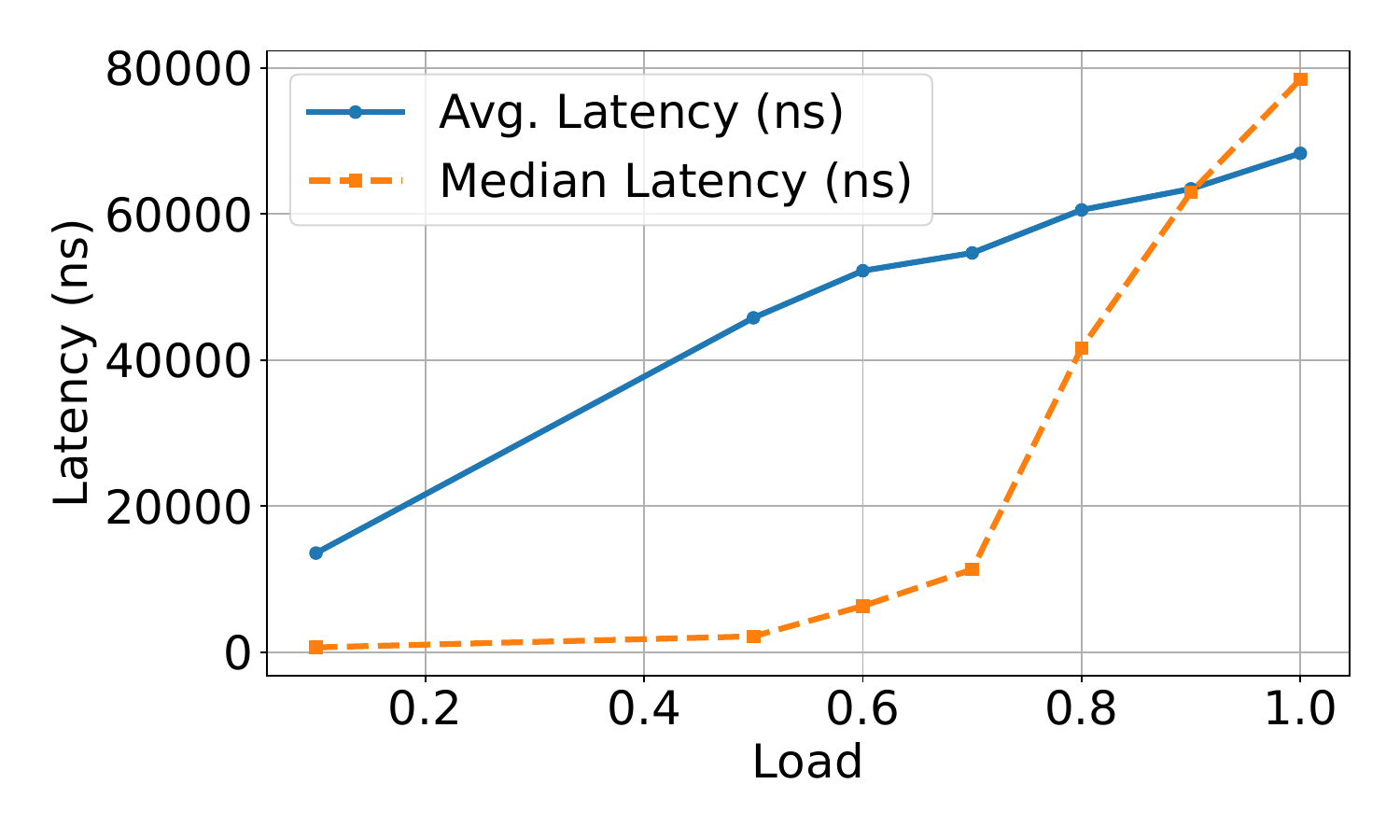}
        \captionof{figure}{\textmd{Characterizing the setup used for simulations.}}
        \label{fig:load_vs_latency}
\end{figure}

\begin{figure}
        \centering
        \includegraphics[width=0.75\linewidth]{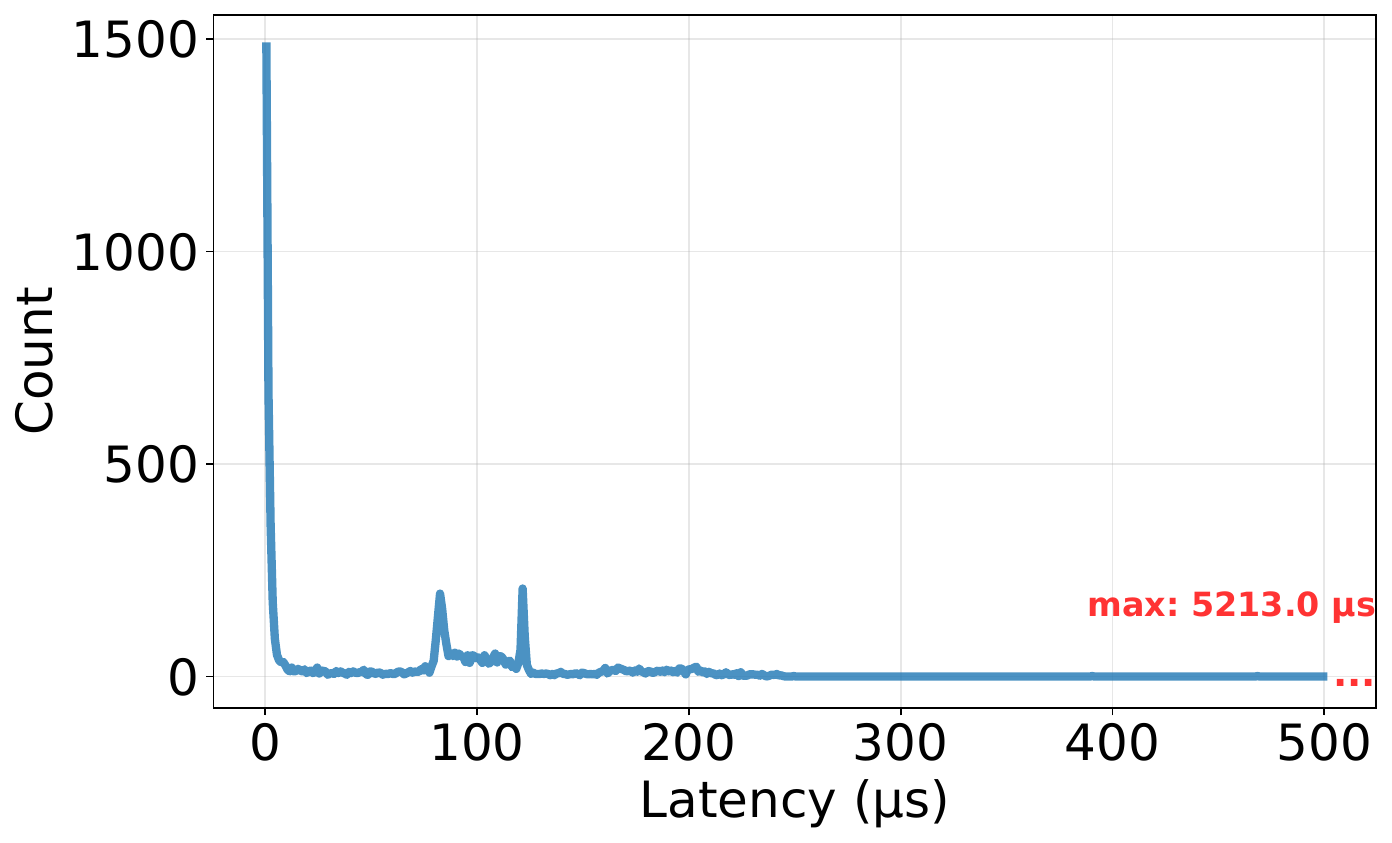}
        \captionof{figure}{\textmd{Histogram of latencies experienced by packets (load=1.0)}}
        \label{fig:dcn_latencies}
\end{figure}

\begin{figure}
    \centering
    \includegraphics[width=0.95\linewidth]{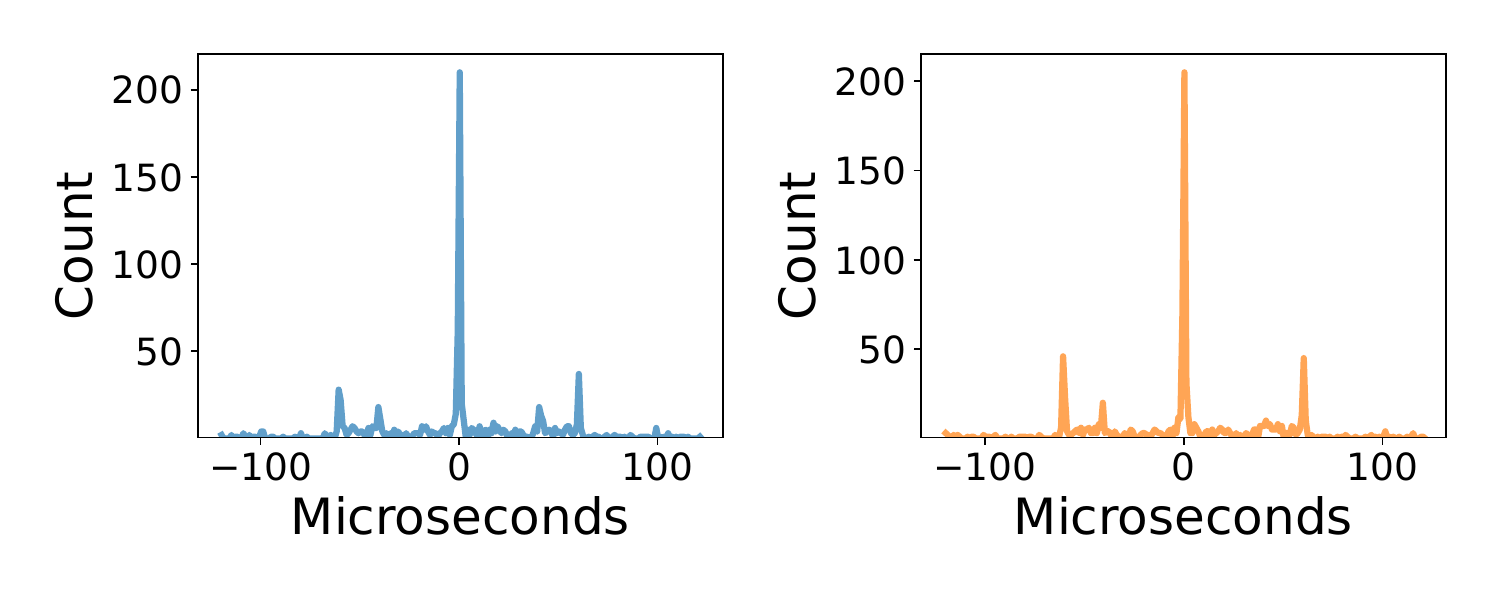}
    \caption{Clock correction distributions for two clients}
    \label{fig:correction_dists}
\end{figure}

% \begin{figure*}
%     \centering

%     \begin{minipage}[t]{0.32\linewidth}
%         \centering
%         \includegraphics[width=\linewidth]{figures/workload_msgsizes_cdf.pdf}
%         \vspace{-0.7cm}
%         \captionof{figure}{\textmd{DCN workload from Homa~\cite{homa, homa_ns3}}}
%         \label{fig:dctcp-workload}
%     \end{minipage}
%     \hfill
%     \begin{minipage}[t]{0.32\linewidth}
%         \centering
%         \includegraphics[width=\linewidth]{figures/latency_vs_load.pdf}
%         \vspace{-0.7cm}
%         \captionof{figure}{\textmd{Characterizing the setup used for simulations.}}
%         \label{fig:load_vs_latency}
%     \end{minipage}
%     \hfill
%     \begin{minipage}[t]{0.32\linewidth}
%         \centering
%         \includegraphics[width=\linewidth]{figures/dcn_latencies.pdf}
%         \vspace{-0.7cm}
%         \captionof{figure}{\textmd{Histogram of latencies experienced by packets (load=1.0)}}
%         \label{fig:dcn_latencies}
%     \end{minipage}

% \end{figure*}

%% file: main.bib
@inproceedings{dbo,
author = {Gupta, Eashan and Goyal, Prateesh and Marinos, Ilias and Zhao, Chenxingyu and Mittal, Radhika and Chandra, Ranveer},
title = {DBO: Fairness for Cloud-Hosted Financial Exchanges},
year = {2023},
isbn = {9798400702365},
publisher = {Association for Computing Machinery},
address = {New York, NY, USA},
url = {https://doi.org/10.1145/3603269.3604871},
doi = {10.1145/3603269.3604871},
booktitle = {Proceedings of the ACM SIGCOMM 2023 Conference},
pages = {550–563},
numpages = {14},
location = {New York, NY, USA},
series = {ACM SIGCOMM '23}
}

@misc{onyx_report,
      title={Design and Implementation of a Scalable Financial Exchange in the Public Cloud}, 
      author={Muhammad Haseeb and Jinkun Geng and Daniel Duclos-Cavalcanti and Ulysses Butler and Xiyu Hao and Anirudh Sivaraman and Srinivas Narayana},
      year={2025},
      eprint={2402.09527},
      archivePrefix={arXiv},
      primaryClass={cs.NI},
      url={https://arxiv.org/abs/2402.09527}, 
}

@inproceedings{onyx,
author = {Haseeb, Muhammad and Geng, Jinkun and Duclos-Cavalcanti, Daniel and Hao, Xiyu and Butler, Ulysses and Mittal, Radhika and Narayana, Srinivas and Sivaraman, Anirudh},
title = {Network Support For Scalable And High Performance Cloud Exchanges},
year = {2025},
isbn = {9798400715242},
publisher = {Association for Computing Machinery},
address = {New York, NY, USA},
url = {https://doi.org/10.1145/3718958.3750530},
doi = {10.1145/3718958.3750530},
booktitle = {Proceedings of the ACM SIGCOMM 2025 Conference},
pages = {1110–1131},
numpages = {22},
location = {S\~{a}o Francisco Convent, Coimbra, Portugal},
series = {SIGCOMM '25}
}

@inproceedings{cloudex,
author = {Ghalayini, Ahmad and Geng, Jinkun and Sachidananda, Vighnesh and Sriram, Vinay and Geng, Yilong and Prabhakar, Balaji and Rosenblum, Mendel and Sivaraman, Anirudh},
title = {CloudEx: A Fair-Access Financial Exchange in the Cloud},
year = {2021},
isbn = {9781450384384},
publisher = {Association for Computing Machinery},
address = {New York, NY, USA},
url = {https://doi.org/10.1145/3458336.3465278},
doi = {10.1145/3458336.3465278},
booktitle = {Proceedings of the Workshop on Hot Topics in Operating Systems},
pages = {96–103},
numpages = {8},
location = {Ann Arbor, Michigan},
series = {HotOS '21}
}

@inproceedings{jasper,
author = {Haseeb, Muhammad and Geng, Jinkun and Butler, Ulysses and Hao, Xiyu and Duclos-Cavalcanti, Daniel and Sivaraman, Anirudh},
title = {POSTER: Jasper, A Scalable and Fair Multicast for Financial Exchanges in the Cloud},
year = {2024},
isbn = {9798400707179},
publisher = {Association for Computing Machinery},
address = {New York, NY, USA},
url = {https://doi.org/10.1145/3672202.3673728},
doi = {10.1145/3672202.3673728},
booktitle = {Proceedings of the ACM SIGCOMM 2024 Conference: Posters and Demos},
pages = {36–38},
numpages = {3},
location = {Sydney, NSW, Australia},
series = {ACM SIGCOMM Posters and Demos '24}
}

@inproceedings {nopaxos,
author = {Jialin Li and Ellis Michael and Naveen Kr. Sharma and Adriana Szekeres and Dan R. K. Ports},
title = {Just Say {NO} to Paxos Overhead: Replacing Consensus with Network Ordering},
booktitle = {12th USENIX Symposium on Operating Systems Design and Implementation (OSDI 16)},
year = {2016},
isbn = {978-1-931971-33-1},
address = {Savannah, GA},
pages = {467--483},
url = {https://www.usenix.org/conference/osdi16/technical-sessions/presentation/li},
publisher = {USENIX Association},
month = nov
}

@inproceedings {hydra,
author = {Inho Choi and Ellis Michael and Yunfan Li and Dan R. K. Ports and Jialin Li},
title = {Hydra: {Serialization-Free} Network Ordering for Strongly Consistent Distributed Applications},
booktitle = {20th USENIX Symposium on Networked Systems Design and Implementation (NSDI 23)},
year = {2023},
isbn = {978-1-939133-33-5},
address = {Boston, MA},
pages = {293--320},
url = {https://www.usenix.org/conference/nsdi23/presentation/choi},
publisher = {USENIX Association},
month = apr
}

@article{paxos,
author = {Lamport, Leslie},
title = {The part-time parliament},
year = {1998},
issue_date = {May 1998},
publisher = {Association for Computing Machinery},
address = {New York, NY, USA},
volume = {16},
number = {2},
issn = {0734-2071},
url = {https://doi.org/10.1145/279227.279229},
doi = {10.1145/279227.279229},
journal = {ACM Trans. Comput. Syst.},
month = may,
pages = {133–169},
numpages = {37}
}

@inproceedings{raft,
author = {Ongaro, Diego and Ousterhout, John},
title = {In search of an understandable consensus algorithm},
year = {2014},
isbn = {9781931971102},
publisher = {USENIX Association},
address = {USA},
booktitle = {Proceedings of the 2014 USENIX Conference on USENIX Annual Technical Conference},
pages = {305–320},
numpages = {16},
location = {Philadelphia, PA},
series = {USENIX ATC'14}
}

@article{spanner,
author = {Corbett, James C. and Dean, Jeffrey and Epstein, Michael and Fikes, Andrew and Frost, Christopher and Furman, J. J. and Ghemawat, Sanjay and Gubarev, Andrey and Heiser, Christopher and Hochschild, Peter and Hsieh, Wilson and Kanthak, Sebastian and Kogan, Eugene and Li, Hongyi and Lloyd, Alexander and Melnik, Sergey and Mwaura, David and Nagle, David and Quinlan, Sean and Rao, Rajesh and Rolig, Lindsay and Saito, Yasushi and Szymaniak, Michal and Taylor, Christopher and Wang, Ruth and Woodford, Dale},
title = {Spanner: Google’s Globally Distributed Database},
year = {2013},
issue_date = {August 2013},
publisher = {Association for Computing Machinery},
address = {New York, NY, USA},
volume = {31},
number = {3},
issn = {0734-2071},
url = {https://doi.org/10.1145/2491245},
doi = {10.1145/2491245},
journal = {ACM Trans. Comput. Syst.},
month = aug,
articleno = {8},
numpages = {22}
}

@misc{shoebot1,
  author       = {Nike Shoe Bot},
  title        = {Nike Shoe Bot - The Ultimate Sneaker Bot},
  year         = {2025},
  url          = {https://www.nikeshoebot.com/},
  note         = {Accessed: 2025-03-19}
}

@misc{shoebot2,
  author       = {{AIO Bot}},
  title        = {AIO Bot | The Ultimate Sneaker Bot for Automatic Copping},
  year         = 2025,
  url          = {https://www.aiobot.com/},
  note         = {Accessed: 2025-03-19}
}

@misc{bot3,
  author       = {{Kasada}},
  title        = {NFT Bots: How They Work and How to Stop Them},
  year         = 2025,
  url          = {https://www.kasada.io/nft-bots/},
  note         = {Accessed: 2025-03-19}
}

@inproceedings {huygens,
author = {Yilong Geng and Shiyu Liu and Zi Yin and Ashish Naik and Balaji Prabhakar and Mendel Rosenblum and Amin Vahdat},
title = {Exploiting a Natural Network Effect for Scalable, Fine-grained Clock Synchronization},
booktitle = {15th USENIX Symposium on Networked Systems Design and Implementation (NSDI 18)},
year = {2018},
isbn = {978-1-939133-01-4},
address = {Renton, WA},
pages = {81--94},
url = {https://www.usenix.org/conference/nsdi18/presentation/geng},
publisher = {USENIX Association},
month = apr
}

@inproceedings{aerial_marketplace,
author = {Balasingam, Arjun and Gopalakrishnan, Karthik and Mittal, Radhika and Alizadeh, Mohammad and Balakrishnan, Hamsa and Balakrishnan, Hari},
title = {Toward a Marketplace for Aerial Computing},
year = {2021},
isbn = {9781450385992},
publisher = {Association for Computing Machinery},
address = {New York, NY, USA},
url = {https://doi.org/10.1145/3469259.3470485},
doi = {10.1145/3469259.3470485},
booktitle = {Proceedings of the 7th Workshop on Micro Aerial Vehicle Networks, Systems, and Applications},
pages = {1–6},
numpages = {6},
location = {Virtual, WI, USA},
series = {Dronet '21}
}

@article{lundelius_clock,
title = {An upper and lower bound for clock synchronization},
journal = {Information and Control},
volume = {62},
number = {2},
pages = {190-204},
year = {1984},
issn = {0019-9958},
doi = {https://doi.org/10.1016/S0019-9958(84)80033-9},
url = {https://www.sciencedirect.com/science/article/pii/S0019995884800339},
author = {Jennifer Lundelius and Nancy Lynch}
}

@ARTICLE{limits_on_clock_sync,
  author={Freris, Nikolaos M. and Graham, Scott R. and Kumar, P. R.},
  journal={IEEE Transactions on Automatic Control}, 
  title={Fundamental Limits on Synchronizing Clocks Over Networks}, 
  year={2011},
  volume={56},
  number={6},
  pages={1352-1364},
  doi={10.1109/TAC.2010.2089210}}

@article{leslie_ordering,
author = {Lamport, Leslie},
title = {Time, clocks, and the ordering of events in a distributed system},
year = {1978},
issue_date = {July 1978},
publisher = {Association for Computing Machinery},
address = {New York, NY, USA},
volume = {21},
number = {7},
issn = {0001-0782},
url = {https://doi.org/10.1145/359545.359563},
doi = {10.1145/359545.359563},
journal = {Commun. ACM},
month = jul,
pages = {558–565},
numpages = {8}
}

@article{efron_dice,
author = {Richard P. Savage Jr. and},
title = {The Paradox of Nontransitive Dice},
journal = {The American Mathematical Monthly},
volume = {101},
number = {5},
pages = {429--436},
year = {1994},
publisher = {Taylor \& Francis},
doi = {10.1080/00029890.1994.11996968},
URL = {https://doi.org/10.1080/00029890.1994.11996968},
eprint = { https://doi.org/10.1080/00029890.1994.11996968}
}

@misc{adex,
  author       = {Amazon Ads},
  title        = {What is Real-Time Bidding (RTB)? Definition and Importance},
  howpublished = {\url{https://advertising.amazon.com/library/guides/real-time-bidding}},
  note         = {Accessed: 2025-04-07}
}

@misc{equal_wire_length,
author = {},
title = {{Beyond Flash Boys: Improving Transparency and Fairness in Financial Markets}},
howpublished = {\url{http://video.cfainstitute.org/services/player/bcpid3577743869001?bckey=AQ~~,AAABE5oc3_E~,Leu10fA0D1sc9Dh9wz3oyrstQJ-PkzpJ&bctid=4436841897001}},
month = {},
year = {},
note = {Accessed: 2021-02-02}
}

@misc{chrony,
author = {},
title = {{chrony}},
howpublished = {\url{https://chrony-project.org/}},
month = {},
year = {},
note = {Accessed: 2026-02-03}
}

@misc{homa_ns3,
author = {},
title = {{Homa}},
howpublished = {\url{https://github.com/serhatarslan-hub/HomaL4Protocol-ns-3}},
month = {},
year = {},
note = {Accessed: 2026-02-03}
}

@inproceedings{homa,
author = {Montazeri, Behnam and Li, Yilong and Alizadeh, Mohammad and Ousterhout, John},
title = {Homa: a receiver-driven low-latency transport protocol using network priorities},
year = {2018},
isbn = {9781450355674},
publisher = {Association for Computing Machinery},
address = {New York, NY, USA},
url = {https://doi.org/10.1145/3230543.3230564},
doi = {10.1145/3230543.3230564},
booktitle = {Proceedings of the 2018 Conference of the ACM Special Interest Group on Data Communication},
pages = {221–235},
numpages = {15},
location = {Budapest, Hungary},
series = {SIGCOMM '18}
}

@misc{cme,
author = {CME},
title = {{Chicago Mercantile Exchange - Futures \& Options Trading}},
howpublished = {\url{https://www.cmegroup.com}},
month = {},
year = {},
note = {Accessed: 2021-05-09}
}

@inproceedings{syncms-2002,
  title={{Sync-MS: Synchronized Messaging Service for Real-Time Multi-Player Distributed Games}},
  author={Lin, Yow-Jian and Guo, Katherine and Paul, Sanjoy},
  booktitle={10th IEEE International Conference on Network Protocols, 2002. Proceedings.},
  pages={155--164},
  year={2002},
  organization={IEEE}
}

@Misc{aws-clock-sync,
  author	= {{AWS}},
  title		= {{Amazon Time Sync Service expands Microsecond-Accurate time to 87 additonal EC2 instance types}},
  howpublished	= {\url{https://aws.amazon.com/about-aws/whats-new/2024/04/amazon-time-sync-service-microsecond-accurate-time-additonal-ec2-instance-types/}},
  year={2024},
  note   = {Accessed: 08/31/2024}
}

@InProceedings{nsdi22-graham,
  author    = {Ali Najafi and Michael Wei},
  title     = {Graham: Synchronizing Clocks by Leveraging Local Clock Properties},
  booktitle = {Proceedings of the 19th USENIX Symposium on Networked Systems Design and Implementation (NSDI 2022)},
  year      = {2022},
  pages     = {453--466},
  publisher = {USENIX Association},
  address   = {Renton, WA, USA},
  month     = apr,
  url       = {https://www.usenix.org/conference/nsdi22/presentation/najafi},
}

@article{condorcet-book,
    author = {Sugden, Robert},
    title = {Condorcet: Foundations of Social Choice and Political Theory.},
    journal = {The Economic Journal},
    volume = {105},
    number = {432},
    pages = {1296-1297},
    year = {1995},
    month = {09},
    issn = {0013-0133},
    doi = {10.2307/2235427},
    url = {https://doi.org/10.2307/2235427},
    eprint = {https://academic.oup.com/ej/article-pdf/105/432/1296/27040852/ej1296.pdf},
}

@article{smithset,
 ISSN = {00129682, 14680262},
 URL = {http://www.jstor.org/stable/1914033},
 author = {John H. Smith},
 journal = {Econometrica},
 number = {6},
 pages = {1027--1041},
 publisher = {[Wiley, Econometric Society]},
 title = {Aggregation of Preferences with Variable Electorate},
 urldate = {2026-01-22},
 volume = {41},
 year = {1973}
}

@Article{ntp,
  author    = {Mills, D. L.},
  title     = {Internet Time Synchronization: The Network Time Protocol},
  journal   = {IEEE Transactions on Communications},
  year      = {1991},
  volume    = {39},
  number    = {10},
  pages     = {1482--1493},
  doi       = {10.1109/26.103043}
}

@inproceedings {osdi20-pompe,
author = {Yunhao Zhang and Srinath Setty and Qi Chen and Lidong Zhou and Lorenzo Alvisi},
title = {Byzantine Ordered Consensus without Byzantine Oligarchy},
booktitle = {14th USENIX Symposium on Operating Systems Design and Implementation (OSDI 20)},
year = {2020},
isbn = {978-1-939133-19-9},
pages = {633--649},
url = {https://www.usenix.org/conference/osdi20/presentation/zhang-yunhao},
publisher = {USENIX Association},
month = nov
}

@article{clock-temperature-report,
  title={Frequency accuracy \& stability dependencies of crystal oscillators},
  author={Zhou, Hui and Nicholls, Charles and Kunz, Thomas and Schwartz, Howard},
  journal={Carleton University, Systems and Computer Engineering, Technical Report SCE-08-12},
  year={2008},
  publisher={Carleton University Ottawa, ON, USA}
}

@article{tarjan-algo,
author = {Tarjan, Robert},
title = {Depth-First Search and Linear Graph Algorithms},
journal = {SIAM Journal on Computing},
volume = {1},
number = {2},
pages = {146-160},
year = {1972},
doi = {10.1137/0201010},
URL = {https://doi.org/10.1137/0201010},
eprint = {https://doi.org/10.1137/0201010},
}

@inproceedings{sigcomm23-neobft,
author = {Sun, Guangda and Jiang, Mingliang and Khooi, Xin Zhe and Li, Yunfan and Li, Jialin},
title = {NeoBFT: Accelerating Byzantine Fault Tolerance Using Authenticated In-Network Ordering},
year = {2023},
isbn = {9798400702365},
publisher = {Association for Computing Machinery},
address = {New York, NY, USA},
url = {https://doi.org/10.1145/3603269.3604874},
doi = {10.1145/3603269.3604874},
booktitle = {Proceedings of the ACM SIGCOMM 2023 Conference},
pages = {239–254},
numpages = {16},
location = {New York, NY, USA},
series = {ACM SIGCOMM '23}
}

@InProceedings{nsdi26-switchbft,
  author    = {Lior Zeno and Naama Ben-David and Mark Silberstein},
  title     = {In Link We Trust: BFT at the Speed of CFT using Switches},
  booktitle = {Proceedings of the 23rd USENIX Symposium on Networked Systems Design and Implementation (NSDI 2026)},
  year      = {2021},
  url       = {https://www.usenix.org/conference/nsdi26/presentation/zeno}
}
